\documentclass[11pt,english]{article}
\usepackage{amssymb, amsmath}
\usepackage{multicol}
\usepackage{color}
\usepackage{amsmath, amssymb}
\usepackage[mathscr]{eucal}
\usepackage{graphicx}
\usepackage{amsfonts}
\usepackage{babel}
\usepackage{framed}
\usepackage{hyperref}
\newcommand{\bfi}{\bfseries\itshape}

\newcommand{\rem}[1]{}
\newcommand{\remfigure}[1]{}

\makeatletter
\@addtoreset{equation}{section}

\makeatother

\newcommand{\comment}[1]{\vspace{1 mm}\par
\marginpar{\large\underline{}}\noindent
\framebox{\begin{minipage}[c]{0.95 \textwidth}
\color{blue}\bfi #1 \end{minipage}}\vspace{2 mm}\par}
\newcommand{\pa}{{\partial}}

\newtheorem{theorem}{Theorem}

\newtheorem{corollary}[theorem]{Corollary}

\newtheorem{definition}[theorem]{Definition}

\newtheorem{proposition}[theorem]{Proposition}

\newtheorem{remark}[theorem]{Remark}

\numberwithin{theorem}{section}
\newenvironment{proof}[1][Proof]{\textbf{#1.} }{\ \rule{0.5em}{0.5em}}

\def\0{{\bf 0}}

\newcommand{\pp}[2]{\frac{\partial #1}{\partial #2}}
\newcommand{\dd}[2]{\frac{d #1}{d #2}}
\newcommand{\dede}[2]{\frac{\delta #1}{\delta #2}}

\begin{document}

\title{\vspace{-3mm}
Geodesic flows on semidirect-product Lie groups:\\
geometry of singular measure-valued solutions\\
\vspace{-.7mm}
}
\author{
Darryl D. Holm$^{1,\,2}$ and Cesare Tronci$^{1,3}\!$\\ \vspace{-.4cm}
\\
{\footnotesize $^1$ \it Department of Mathematics, Imperial College London,
180 Queen's Gate, London SW7 2AZ, UK}\\
{\footnotesize $^2$ \it Institute for Mathematical Sciences, Imperial College
London, 53 Prince's Gate, London SW7 2PG, UK
} \\
{\footnotesize $^3$\,\it TERA Foundation for Oncological Hadrontherapy,
11 V. Puccini, Novara 28100, Italy}
\\ 
}

\date{\vspace{-.2cm}\normalsize\today}

\maketitle
\vspace{-1cm}
\begin{abstract}\noindent
The EPDiff equation (or dispersionless Camassa-Holm equation  in 1D) is a well known example of geodesic motion on the Diff group of smooth invertible maps (diffeomorphisms). Its recent two-component extension governs geodesic motion on the semidirect product ${\rm Diff}\,\circledS\,{\cal F}$, where $\mathcal{F}$ denotes the space of scalar functions. This paper generalizes the second construction to consider geodesic motion on ${\rm Diff}\,\circledS\mathfrak{\,g}$, where $\mathfrak{g}$ denotes the space of scalar functions that take values on a certain Lie algebra (for example, $\mathfrak{g}=\mathcal{F}_{\,} \otimes_{\,}\mathfrak{so}(3)$). Measure-valued delta-like solutions are shown to be momentum maps possessing a dual pair structure, thereby extending previous results for the EPDiff equation. The collective Hamiltonians are shown to fit into the Kaluza-Klein theory of particles in a Yang-Mills field and these formulations are shown to apply also at the continuum PDE level. In the continuum description, the Kaluza-Klein approach produces the Kelvin circulation theorem.
\end{abstract}

\vspace{-.4cm}

\tableofcontents

\bigskip

\section{Introduction}\label{sec.1}

\subsection{Singular solutions in continuum mechanics}
Singular measure-valued solutions arise in the study of many continuum systems. A famous example of singular solutions in ideal fluids is the point vortex solution for the Euler vorticity equation on the plane. Point vortices are delta-like solutions that follow a multi-particle dynamics. In three dimensions one extends this concept to vortex filaments or vortex sheets, for which the vorticity is supported on a lower dimensional submanifold (1D or 2D respectively) of the Euclidean space $\mathbb{R}^3$. These solutions form an invariant manifold. However, they are not expected to be created by fluid motion from smooth initial conditions.
Singular solutions also exist in plasma physics as magnetic vortex lines and in kinetic theory as single-particle solutions, see, e.g. \cite{GiHoTr07} for discussion and references.

Another well-known fluid model admitting singular solutions is the Camassa-Holm (CH) equation, which is an integrable Hamiltonian fluid equation describing shallow-water solitons moving in one dimension  \cite{CaHo1993}.  In the dispersionless limit, the CH solitons are singular solutions whose velocity profile has a sharp peak at which the derivative is discontinuous. These are the {\it peakon} solutions of the CH equation. In higher dimensions the dispersionless CH equation is often called EPDiff, which is short for {\it Euler-Poincar\'e equation on the diffeomorphisms} \cite{HoMa2004,HoMaRa}. EPDiff also has applications in other areas such as turbulence \cite{FoHoTi01} and imaging \cite{HoRaTrYo2004,HoTrYo2007}.
EPDiff and the Euler fluid equation share the important property of being geodesic motions on diffeomorphism groups. In particular, EPDiff involves the entire Diff group, while its restriction to  the volume-preserving transformations Diff$_{\rm vol}$ yields the Euler fluid equation. Moreover, the EPDiff equation has the additional interesting feature of showing \emph{spontaneous} emergence of singular solutions from any confined initial velocity configuration.
In 1D this result follows from the \emph{steepening lemma} \cite{CaHo1993}. The dynamical variable is the fluid velocity (or momentum) and the peaks of the singular solutions follow the Lagrangian trajectories of fluid particles. Recent work has proven that the singular solutions of EPDiff represent a certain kind of momentum map \cite{HoMa2004}. This result answers the fundamental question  concerning the geometric nature of these solutions. In particular, they are equivariant momentum maps admitting a dual pair representation similar to that found for vortex filaments in \cite{MaWe83}.

The Camassa-Holm equation in 1D also possesses a two-component integrable extension (CH2) \cite{ChLiZh2005,Falqui06,Ku2007}, which represents geodesic motion on the semidirect product ${\rm Diff}\,\circledS\,{\cal F}$, where $\mathcal{F}$ denotes the space of real scalar functions.
This system of equations involves both fluid density and momentum, and it also possesses singular solutions in the latter variable. CH peakon solutions form an invariant subspace of CH2 solutions for the case that the density vanishes identically. A change in the metric allows delta-like singularities in {\it both} variables, not just the fluid momentum, although such a change may also destroy integrability \cite{HoTrYo2007}.
This paper investigates the dynamics and the geometric nature of singular solutions of the geodesic equations on semidirect product Lie groups. In particular, it shows that these solutions are still momentum maps; so that all the results for momentum maps in the case of EPDiff also hold for this system.

\subsection{The Camassa-Holm and EPDiff equations}
The dispersionless Camassa-Holm equation is a geodesic flow on the infinite-dimensional Lie group of smooth invertible maps (diffeomorphisms) of either the real line or a periodic interval. This geodesic flow exhibits spontaneous emergence of singularities from \emph{any} confined smooth initial velocity profile. The CH equation is a 1+1 partial differential equation (PDE) for the 1D fluid velocity vector $u$ as a function of position $x\in\mathbb{R}$ and time $t\in\mathbb{R}$ and is written as  \cite{CaHo1993}
\begin{equation}
u_t+2\kappa u_x-u_{xxt}+3uu_x=2u_x u_{xx}+u u_{xxx}
\,.
\label{CHeqn}
\end{equation}
The present work focuses on the CH equation (\ref{CHeqn}) in the dispersionless case for which  $\kappa=0$ and considers periodic boundary conditions or sufficiently rapid decay at infinity, so that boundary terms do not contribute in integrations by parts. Being geodesic, the CH equation can be recovered from Hamilton's principle with a quadratic Lagrangian, 
\[
\delta\!\int_{t_0}^{t_1}\!\!L(u)\,{\rm d}t=0
\qquad\text{ with }\qquad
L(u)=\frac12\int \!u(x) \,(1-\partial_x^2)\,u(x)\,\,{\rm d}x
\,.
\]
In particular, it follows from the
Euler-Poincar\'e variational principle defined on $\mathfrak{X}(\mathbb{R})=T_e{\rm Diff}(\mathbb{R})$, the Lie algebra of the Diff group, consisting of vector fields on the line \cite{HoMaRa}.
The corresponding Lie-Poisson Hamiltonian formulation follows from the Legendre transform
\[
m=\frac{\delta L}{\delta u}=u-u_{xx}
\quad
\Rightarrow
\quad
u=(1-\partial_x^2)^{-1} m
\,,
\]
in which $m\in\mathfrak{X}^*(\mathbb{R})$. Here $\mathfrak{X}^*(\mathbb{R})$ is the space of one-form densities, which are dual  to the vector fields on the line under the $L^2$ pairing.
After the Legendre transformation, the Hamiltonian becomes
\[
H(m)=\frac12\int\!m\,(1-\partial_x^2)^{-1}m\,{\rm d}x
\quad
\Rightarrow
\quad
u=\frac{\delta H}{\delta m}
\]
and the corresponding Lie-Poisson form of equation (\ref{CHeqn}) with $\kappa=0$ is \cite{CaHo1993}
\[
m_t
+um_x +2u_xm
=0
\,.
\]
The main result for this equation is its complete integrability, which is guaranteed by its bi-Hamiltonian structure \cite{CaHo1993}. Another important feature, called the \emph{steepening lemma} \cite{CaHo1993}, provides the mechanism for the spontaneous emergence of the singular solutions (the peaked solitons, or \emph{peakons} mentioned earlier) from any confined initial velocity distribution. 

\paragraph{EPDiff.}
Except for integrability, these one-dimensional results can be generalized to two or three dimensions, in which the equation becomes EPDiff, namely,
\begin{equation}
\partial_t\mathbf{m}
- \mathbf{u}\times{\rm curl\,}\mathbf{m}
+ \nabla (\mathbf{u}\cdot\mathbf{m})
+ \mathbf{m}({\rm div\,}\mathbf{u})
=0
\,.
\label{EPDiff-3D}
\end{equation}
The EPDiff Hamiltonian on $\mathfrak{X}^*(\mathbb{R}^3)$ arising from this generalization is given by
\begin{equation}
H({\bf m})
=\frac12\int\!{\bf m}\cdot(1-\alpha^2\Delta)^{-1}{\bf m}\,{\rm
d}{\bf x}
=:{\frac12}\|{\bf m}\|^2
\,,
\label{Ham-EPDiff}
\end{equation}
in which $\alpha$ is the length-scale over which the velocity is smoothed relative to the momentum via the relation
\begin{equation}
{\bf u}=(1-\alpha^2\Delta)^{-1}{\bf m}
\,.
\label{vel-EPDiff}
\end{equation}
The singular solutions in higher dimensions are written in the momentum representation
as
\begin{equation}
{\bf m(x},t)
=\sum_{i=1}^N\int{\bf P}_{\!i}(s,t)\,\delta({\bf x-Q}_i(s,t))\,{\rm d}s
\,,
\label{singsolepdiff}
\end{equation}
where $s$ is a variable of dimension $k<3$. These solutions represent moving filaments or sheets, when $s$ has dimension 1 or 2, respectively.

A further generalization replaces the kernel that defines the norm of $\|{\bf m}\|$ in (\ref{Ham-EPDiff}) via the relation (\ref{vel-EPDiff}) with the general convolution
\begin{equation}
G\,*\,{\bf m}=\int G({\bf x-x'})\,\bf m(x')\,{\rm d}{\bf x}'
\,,
\label{vel-pulson}
\end{equation}
involving an arbitrary Green's function, or kernel, $G$. This generalization produces the metric,
\begin{equation}
H({\bf m})
=\frac12\int\!{\bf m}\cdot (G\,*\,{\bf m})
\,{\rm d}{\bf x}
=:{\frac12}\|{\bf m}\|_G^2
\,,
\label{Ham-EPDiff-invmetric}
\end{equation}
which is the Lie-Poisson Hamiltonian for EPDiff (\ref{EPDiff-3D}) and is, thus, invariant under its evolution.
The dynamics of $({\bf Q}_i,{\bf P}_{\!i})$ with $i=1,\dots,N$ for the  singular pulson solutions in (\ref{singsolepdiff}) is given by canonical Hamiltonian dynamics with the Hamiltonian
\[
\mathcal{H}=\frac12\sum_{i,j}\iint {\bf P}_{\!i}(s,t)\cdot{\bf P}_{\!j}(s',t)\,\,G({\bf Q}_i(s,t)-{\bf Q}_j(s',t))\,\,{\rm d}s\,{\rm d}s'
\,,
\]
obtained by evaluating the Lie-Poisson Hamiltonian for EPDiff on its singular solution set.
\begin{definition}[Momentum map \cite{MaRa99}]
Given a Poisson manifold (i.e. a manifold $P$ with a Poisson bracket  $\{\cdot,\cdot\}$ defined on the functions $\mathcal{F}(P)$)
and a Lie group $G$ acting on it by Poisson maps (that is, via the Poisson structure), a momentum map $\mathbf{J}:P\rightarrow\mathfrak{g}^*$ is defined by the following formula,
\[
\big\{F(p),\langle \mathbf{J}(p),\xi\rangle\big\}=\xi_P[F(p)]
\qquad\quad\forall\, F\in\mathcal{F}(P),\quad\forall\,\xi\in\mathfrak{g}
\,.
\]
Here $\mathcal{F}(P)$ denotes the functions on $P$, $\mathfrak{g}$ is the Lie algebra of $G$ and $\xi_P$ is the vector field given by the infinitesimal generator of the action of $G$ on the manifold $P$ by Poisson maps.
\end{definition}

\begin{theorem}
The singular solution (\ref{singsolepdiff}) is a \emph{momentum map} \cite{HoMa2004}.
\end{theorem}

To illustrate this theorem, fix a $k$-dimensional manifold $S$ immersed in $\mathbb{R}^n$ and consider the embedding ${\mathbf{Q}_i:S\rightarrow\mathbb{R}^n}$. Such embeddings form
a smooth manifold $\text{Emb}(S,\mathbb{R}^n)$ and thus one can consider its cotangent bundle $(\mathbf{Q}_i,\mathbf{P}_i)\in
T^*\text{Emb}(S,\mathbb{R}^n)$. Consider $\text{Diff}(\mathbb{R}^n)$
acting on $\text{Emb}(S,\mathbb{R}^n)$ on the left by composition of functions $\left(g\,\mathbf{Q}=g\circ\mathbf{Q}\right)$
and lift this action to $T^*\text{Emb}(S,\mathbb{R}^n)$. This procedure constructs the singular solution momentum map for EPDiff,
\[
\mathbf{J}:T^*\text{Emb}(S,\mathbb{R}^n)\rightarrow\mathfrak{X}^*(\mathbb{R}^n)
\qquad\text{with}\qquad
\mathbf{J}(\mathbf{Q},\mathbf{P})=\!\int\mathbf{P}(s,t)\,\delta(\mathbf{x}-\mathbf{Q}(s,t))\,ds
\,.
\]
This construction was extensively discussed in \cite{HoMa2004}, where proofs were given in various cases. A key result is that the momentum map constructed this way is {\it equivariant}, which means it is also a Poisson map. This explains why the coordinates $(\mathbf{Q},\mathbf{P})$ undergo Hamiltonian dynamics.
There is also a right action $\left(\mathbf{Q}\,g=\mathbf{Q}\circ g\right)$, whose momentum map corresponds to the canonical one-form ${\bf P}\cdot d{\bf Q}$ on $T^*\text{Emb}$.

The EPDiff equation has been applied in several contexts, ranging from turbulence modeling \cite{FoHoTi01} to the design of imaging techniques \cite{HoRaTrYo2004,HoTrYo2007}. Its CH form in one dimension has been widely studied because of the integrability of its singular solutions  ({\em peakons}).
The properties of geodesic flows on the diffeomorphisms regarded as a Lie group have played a central role in determining the behavior of the singular solutions of EPDiff. This suggests that we pursue further investigations of the emergence of singularities in the solutions of geodesic equations on Lie groups.

\subsection{The two-component Camassa-Holm system}

In recent years, the Camassa-Holm equation has been extended \cite{ChLiZh2005,Falqui06,Ku2007} in order to combine the integrability property with compressibility, which introduces a pressure term in the equation for the fluid momentum. The resulting system (CH2) is a geodesic motion equation on Diff$\,\circledS\,\mathcal{F}$, given as an Euler-Poincar\'e equation on the semidirect product Lie algebra $\mathfrak{X}^{\,}\circledS\,\mathcal{F}$. In the general case, the Euler-Poincar\'e equations are written on the dual of a semidirect product Lie algebra $\mathfrak{g}^{\,}\circledS\,V$ as \cite{HoMaRa}
\begin{align*}
\frac{d}{dt}\frac{\delta L}{\delta (\xi,\,a)}=-\,{\rm ad}^*_{(\xi,\,a)\,}\frac{\delta L}{\delta (\xi,\,a)}
\qquad
(\xi,\,a)\in \mathfrak{g}^{\,}\circledS\,V
\end{align*}
whose components are
\begin{align*}
\frac{d}{dt}\frac{\delta L}{\delta \xi}=-\,{\rm ad}^*_{\xi\,}\frac{\delta L}{\delta \xi}+\frac{\delta L}{\delta a}\diamond a
\,,\qquad
\frac{d}{dt}\frac{\delta L}{\delta a}=-\,\xi\,\frac{\delta L}{\delta a}
\end{align*}
where the notation $\xi\,{\delta L}/{\delta a}$ stands for the (left) Lie algebra action of $\mathfrak{g}$ on $V^*$.

The integrable CH2 equations are derived from the following variational principle
on $\mathfrak{X}^{\,}\circledS\,\mathcal{F}$
\[
\delta\!\int_{t_0}^{t_1}\!\!L(u,\rho)\,{\rm d}t=0
\qquad\text{ with }\qquad
L(u)=\frac12\int \!u \,(1-\partial_x^2)u\,\,{\rm d}x+\frac12\int \!\rho^2\,\,{\rm d}x
\,.
\]
Explicitly, the CH2 equations are
\begin{align*}
\rho_t=&-(\rho u)_x
\,,\\
u_t -u_{xxt}=&-3uu_x+2u_x u_{xx}+u u_{xxx}-\rho\rho_x
\,.
\end{align*}
These equations describe geodesic motion with respect to the $H^1$ metric in $u$ and the
$L^2$ metric in $\rho$. Legendre transforming yields the metric Hamiltonian
\[
H(m,\rho)=
\frac12\int\!m \,(1-\partial^2)^{-1}\,m\,{\rm d}x +
\frac12\int\!\rho^2\,{\rm d}x
\,.
\]
Extending to more general metrics (Green's functions) yields the Hamiltonian
\[
H(m,\rho)=
\frac12\iint\!m(x) \,G_1(x-x')\,m(x')\,{\rm d}x\,{\rm d}x'
+
\frac12\iint\!\rho(x) \,G_2(x-x')\,\rho(x')\,{\rm d}x\,{\rm d}x'
\,.
\]
This Hamiltonian yields the Lie-Poisson equations
\begin{align*}
\rho_t=&-(\rho u)_x
\,,
\\
m_t=&-u m_x-2mu_x-\rho \lambda_x
\,,
\end{align*}
in which $u$ and $\lambda $ are defined by the variational derivatives,
\[
u=\dede{H}{m}=G_1*m
\quad\hbox{and}\quad
\lambda=\dede{H}{\rho}=G_2*\rho
\,,
\]
and the symbol $*$ denotes convolution.
For example, we may write the Euler-Poincar\'e equations for the $H^1$ norms associated to the two length-scales $\alpha_1$ and $\alpha_2$  (i.e., $G_i=(1-\alpha^2_i\partial^2)^{-1}$) as:
\begin{align*}
\lambda_t-\alpha_2^2\lambda_{xxt}&=-\left(u\lambda-\alpha_2^2 u \lambda_{xx}\right)_x
\,,\\
u_t-\alpha_1^2 u_{xxt}&=-3uu_x+2\alpha_1^2u_x u_{xx}+\alpha_1^2u u_{xxx}
-\lambda_x\left(\lambda-\alpha_2^2\lambda_{xx}\right)
\,.
\end{align*}
The second equation may also be expressed in hydrodynamic form as
\[
u_t+u u_x = -\,p_x
\,,\quad\hbox{with pressure}\quad
p: = G_1*\left(u^2+\frac{\alpha_1^2}{2}u_x^2+\frac12\lambda^2-\frac{\alpha_2^2}{2}\lambda_x^2\right)
\,.
\]
The integrable case CH2 is recovered from these equations for $\alpha_2=0$, although this case does not allow for singular delta-like solutions in the density variable  (or depth variable, in the shallow water wave interpretation), since $\lambda=\rho$ when $\alpha_2=0$.

One may also extend the problem to higher dimensions by taking the Hamiltonian
\begin{equation}\label{EPGosH-Ham}
H({\bf m},\rho)=
\frac12\iint\!{\bf m}({\bf x}) \,G_1({\bf x-x}')\,{\bf m}({\bf x}')\,{\rm d}^n{\bf x}\,{\rm d}^n{\bf x}'+
\frac12\iint\!\rho({\bf x})\,G_2({\bf x-x}')\,\rho({\bf x}')\,{\rm d}^n{\bf x}\,{\rm d}^n{\bf x}'
\end{equation}
which represents the two-component extension of EPDiff and is denoted by EP(Diff$\,\circledS\,\mathcal{F}$).
In three dimensions, the corresponding Lie-Poisson equations assume the form
\begin{align*}
\rho_t=&\,-\nabla\cdot(\rho {\bf u})
\,,\\
\mathbf{m}_t=\,&
 \mathbf{u}\times{\rm curl\,}\mathbf{m}
- \nabla (\mathbf{u}\cdot\mathbf{m})
- \mathbf{m}({\rm div\,}\mathbf{u})
-\rho\nabla\lambda
\,,
\end{align*}
where ${\bf u}=G_1*{\bf m}$ and $\lambda=G_2*\rho$. 
These expressions can also be written in a covariant form by
using the Lie derivative $\pounds_u$ with respect to the velocity vector field $u\in \mathfrak{X}(\mathbb{R}^n)$ of a one-form density
\[
m
=\mathbf{m}\cdot{\rm d}{\bf x}\otimes{\rm d}^n{\bf x}
\in \mathfrak{X}^*(\mathbb{R}^n)
\,,
\]
whose co-vector components are ${\bf m}$.
The {\bfi diamond operator} $(\, \diamond\,)$ corresponding to the Lie derivative $\pounds_u$ is defined by
\begin{equation}
\Big\langle \rho,-\,\pounds_{\bf u} \lambda\Big\rangle
_{V\times V^*}
:=
\Big\langle\rho\diamond \lambda,\bf u\Big\rangle
_{\mathfrak{X}^*\times \mathfrak{X}}
\,,
\label{diamond.def}
\end{equation}
where in the present case $\lambda\in V^*=\mathcal{F}$ is a scalar function while $\rho\in V={\rm Den}$ is a density variable and $\langle\, \cdot\,,\cdot\, \rangle $ is $L^2$ pairing in the corresponding spaces. In this notation, one writes the general case for $\rho\in V$ and $\lambda\in V^*$ as
\begin{align}\nonumber
\rho_t+\pounds_{u}\,\rho
&=
0
\,,\\
{m}_t+\pounds_{u}\,{m}
&=
\rho\diamond\lambda
\,.
\label{EPDiffDen}
\end{align}
Upon recalling that the metrics $G_i$ are Green's functions associated to the differential operators $Q_i$ (that is, $Q_i\cdot G_i=\delta(x-x')$),
one finds the following Lagrangian
\[
L({\bf u},\lambda)=\frac12\int \!{\bf u} \ Q_1 {\bf u}\,\,{\rm d}^n{\bf x}+\frac12\int \!\lambda\
 Q_2\,\lambda\,{\rm d}^n{\bf x}
\]
whose Euler-Poincar\'e equations recovers the EP(Diff$\,\circledS\,\mathcal{F}$) equations (\ref{EPDiffDen}),
upon substituting the Legendre transforms  $\mathbf{m}=Q_1\, {\bf u}$ and $\rho=Q_2\,\lambda$.

Specializing the operator $Q_2$ to $Q_2=\Delta$ (that is, taking $\alpha_2\to\infty$ above) allows this system to be interpreted physically as an incompressible charged fluid. In this case, $\rho$ is interpreted as charge density, rather than mass density, and $\lambda$ is the electrostatic interaction potential, satisfying $\Delta \lambda =\rho$ for the Coulomb potential. The pure CH2 case ($Q_2=1$, or $\alpha_2\to0$) corresponds to a delta-like interaction potential, which also appears in the integrable Benney system \cite{Be1973,Gi1981}. The application of the full semidirect-product framework to the imaging method called {\bfi metamorphosis} was studied in \cite{HoTrYo2007}. The relation of the latter work with the current investigation will be traced further in the sections below.

\subsection{Plan of the paper} This paper follows \cite{HoMa2004} in studying the singular solutions of geodesic flows on semidirect-product Lie groups by identifying them with momentum maps. The new features here are: (i) we explain their geometric nature in the context of the Kaluza-Klein formulation of a particle in a Yang-Mills field and (ii) we extend the applications of this approach to systems of equations governing geodesic flows on semidirect-product Lie groups with a non-Abelian gauge group.

The next section first shows how this works when considering the Abelian gauge group of functions $\cal F$. Momentum maps for left and right actions of the diffeomorphisms and the Lie group of gauge symmetry are derived that recover the corresponding conservation laws. The key observation is that the collective Hamiltonian for the singular solution momentum map is a Kaluza-Klein Hamiltonian, thereby recovering the conservation of the gauge charge.
In section 3 we show how this observation extends to the consideration of a non-Abelian gauge group $G$ for particles carrying a spin-like variable. The introduction of a principal $G$-bundle becomes important in this context, since it best clarifies the
reduction processes that are involved in the collective Hamiltonian system.
Again, the collective Hamiltonian arises from a Kaluza-Klein formulation, after reduction by the gauge symmetry. We show that the Kaluza-Klein construction allows one to extend all the geometric features found in \cite{HoMa2004} for the EPDiff equation to the semidirect-product geodesic flows under consideration.

The last section is devoted to illustrating how the continuum geodesic equations on semidirect products also arise naturally from a Kaluza-Klein formulation. In this formulation, the advection equation of the gauge charge density is recovered as the conservation of the momentum conjugate to a cyclic variable. The Kaluza-Klein approach  also leads to a Kelvin circulation theorem in the continuum description.

\section{The singular solution momentum map for EP(Diff$\,\circledS\boldsymbol{\mathcal{F}}$)}\label{sec.2}

Given the semidirect-product Lie-Poisson equations (\ref{EPDiffDen}), direct substitution shows that they allow the singular solutions
\begin{equation}
\big({\bf m},\rho\big)=\sum_{i=1}^N\int\!\big({\bf P}_i(s,t),w_i(s)\big)\,\delta\!\left({\bf
x-Q}_i(s,t)\right)\,{\rm d}^ks
\,,
\label{SDsingsoln}
\end{equation}
where $s$ is a coordinate on a submanifold $S$ of $\mathbb{R}^n$, exactly as in the case of EPDiff. If dim$\,S=1$, then this corresponds to fluid variables supported on a filament, while dim$\,S=2$ yields sheets of fluid density and momentum. The dynamics of $({\bf Q}_i,{\bf P}_i, w_i)$ is given by
\begin{align*}
\pp{{\bf Q}_i(s,t)}{t}=&\ \sum_j\int\!{\bf P}_j(s',t)\, G_1({\bf Q}_i(s,t)-{\bf Q}_j(s',t))\ {\rm d}^ks'
\,,\\
\pp{{\bf P}_i(s,t)}{t}=&\ -\sum_j\int\!
{\bf P}_i(s,t)\cdot{\bf P}_j(s',t)\,\text{\large$\nabla$}_{\!{\bf Q}_i}G_1({\bf Q}_i(s,t)-{\bf Q}_j(s',t))\ {\rm d}^ks'
\\
&\
- \sum_j\int\! w_i(s)\,w_j(s')\,
\text{\large$\nabla$}_{\!{\bf Q}_i}G_2({\bf Q}_i(s,t)-{\bf Q}_j(s',t))\
{\rm d}^ks'
\,,
\end{align*}
with $\partial_t{w}_i(s)=0,\ \forall i$.

Recalling the geometric nature of the pulson solution of EPDiff and following the reasoning in \cite{HoMa2004},
one interprets ${\bf Q}_i$ as a smooth embedding in Emb$(S,\mathbb{R}^n)$ and $P_i={\bf P}_i\cdot{\rm d}{\bf Q}_i$ (no sum) as the canonical one-form on $T^*$Emb$(S,\mathbb{R}^n)$ for the $i$-th pulson.
In the case of EP(Diff$\,\circledS\,\mathcal{F}$), the weights $w_i$ for $i=1,\dots,N$ are considered as maps $w_i:S\to\mathbb{R}^*$. That is, the weights $w_i$ are  distributions on $S$, so that $w_i\in{\rm
Den}(S)$, where ${\rm Den}:=\mathcal{F}^*$. In particular, we consider the triple
\[
({\bf Q}_i,{\bf P}_i, w_i)\ \, \text{\large $\in$} \ \,
\rem{ 
\bigcup_{{\bf Q}_i\in{\rm Emb}}\!\left(\,T_{{\bf Q}_i}^*{\rm Emb}(S,\mathbb{R}^n)\,\bigoplus\,{\rm Den}(S)\, \right)\ =:\
}    
T^*{\rm Emb}(S,\mathbb{R}^n)\,\times\,{\rm Den}(S)
\,.
\]
and prove the following.


\begin{theorem}[Singular solution momentum map]
\label{singsolnmommap-thm}
The singular solutions (\ref{SDsingsoln}) of the semidirect-product Lie-Poisson equations (\ref{EPDiffDen}) are given by
\[
\big({\bf m},\rho\big)=\sum_{i=1}^N\int\!\big({\bf P}_i(s,t),w_i(s)\big)\,\delta\!\left({\bf
x-Q}_i(s,t)\right)\,{\rm d}^ks
\,.
\]
These expressions for $({\bf m},\rho)\in\mathfrak{X}^*(\mathbb{R}^n)\ \text{\large $\circledS$}\ {\rm Den}(\mathbb{R}^n)$ identify a momentum map
\[
{\bf J}:\underset{i=1}{\overset{N}{\text{\LARGE$\times$}}}\Big(T^*{\rm Emb}(S,\mathbb{R}^n)\,\times\,{\rm Den}(S)\Big)
\,\rightarrow\,
\mathfrak{X}^*(\mathbb{R}^n)\ \text{\large $\circledS$}\ {\rm Den}(\mathbb{R}^n)
\,.
\]
\end{theorem}
\begin{proof}
For convenience, we fix label $i$ and suppress the summations in our singular solution ansatz. In order to define a momentum map we first need to establish a Poisson structure on $T^*{\rm Emb}\,\times\,{\rm Den}$. Since the weights have no temporal evolution, it is reasonable to propose the canonical Poisson bracket on the new phase space, so that
\[
\{F,G\}({\bf Q},{\bf P}, w)\,:=\sum_{j=1}^n
\int
\left(\pp{F}{Q^j}\pp{G}{P_j}-\pp{F}{Q^j}\pp{G}{P_j}\right)
\,{\rm d}^ks
\,.
\]
Now, if $\beta=(\boldsymbol{\beta}_1,\beta_0)\in\mathfrak{X}(\mathbb{R}^n)\ \text{\large $\circledS$}\ \mathcal{F}(\mathbb{R}^n)$, then the pairing $\langle{\bf J},\beta\rangle$ is naturally written as
\begin{equation}
\Big\langle{\bf J}({\bf Q},{\bf P}, w),\,\beta \Big\rangle
=\int
\Big( {\bf P}(s)\cdot\boldsymbol{\beta}_1({\bf Q}(s))
+
w(s)\,\beta_0({\bf Q}(s))
\Big)
\,{\rm d}^ks
\,.
\label{mommap-formula}
\end{equation}
Consequently, one calculates
\begin{align}
\big\{
F,\langle{\bf P},\boldsymbol{\beta}_1({\bf Q})\rangle
\big\}
+
\big\{
F,\langle w,\,\beta_0({\bf Q})\rangle
\big\}
\,=\!
\int\!\left(
\frac{\delta F}{\delta \bf Q}\cdot\boldsymbol{\beta}_1({\bf Q}(s))
-
\left(
\frac{d \boldsymbol{\beta}_1}{d{\bf Q}}^T\cdot\, {\bf P}
+\,w\,\frac{d \beta_0}{d {\bf Q}}
\right)\cdot
\frac{\delta F}{\delta {\bf P}}
\right)
\,{\rm d}^ks
\,.
\label{HamVecFld2}
\end{align}
This may be written equivalently as
\[
\big\{
F, \big\langle{\bf J},\,\beta \big\rangle
\big\}
=X_\beta[F]
\,,
\]
in which the vector field $X_\beta$ has $({\bf Q},{\bf P}, w)$ components
\begin{equation}\label{HamVF}
X_\beta:=
\bigg(\boldsymbol{\beta}_1({\bf Q})\,,
-\,
\left(
\frac{d \boldsymbol{\beta}_1}{d{\bf Q}}^T\cdot\, {\bf P}
+\,w\,\frac{d \beta_0}{d {\bf Q}}
\right)
,\,
0\bigg)
\,.
\end{equation}
This vector field is identified with a Hamiltonian vector field corresponding to the Hamiltonian
\[
H
=\int \Big( w(s)\ \beta_0({\bf Q}(s))\ +\ {\bf P}(s)\cdot
\boldsymbol{\beta}_1({\bf Q}(s))\Big)\, {\rm d}^ks
= \Big\langle{\bf J},\beta \Big\rangle
\,.
\]
This Hamiltonian corresponds to compositions of cotangent lifts $T^*$Diff generated by $\boldsymbol{\beta}_1$ with fiber translations
\[
\text{\large$\tau$}_{\!-{\rm d}(w\,\beta_0)}\cdot (q,p) := (q,p-w\,{\rm d} \beta_0)
\]
generated by $-w\beta_0$ (notice that $w$ is independent of $q$, so that ${\rm d}(w\,\beta_0)=w\,{\rm d}\beta_0)$. Thus $X_\beta$
is an infinitesimal generator.
\end{proof}

\begin{remark}
Theorem \ref{singsolnmommap-thm} proves that the singular solution is a momentum map deriving from the infinitesimal action of ${\rm Diff}\,\circledS\,\mathcal{F}$. However, it does not identify the corresponding {\bfi global group action}. This is explained in the next section below.
\end{remark}

\subsection{A left group action for the singular solution momentum map}
The proof of Theorem \ref{singsolnmommap-thm} shows that the momentum map $\bf J$ in equation (\ref{mommap-formula}) is obtained by the following left action of the semidirect-product group ${\rm Diff}\,\circledS\,{\mathcal{F}}$
\begin{align*}
\Big({\bf Q}^{\,(t)} ,\,{\bf P}^{\,(t)},\, w^{\,(t)}\Big)&=
\left(\eta_{t}\big({\bf Q}^{\,(0)}\big) ,\, {\bf P}^{\,(0)}\cdot^{\,} T_{\,}\eta_t^{-1}\big({\bf Q}^{\,(0)}\big)-\textrm{\large d}\big(w\,\beta_0\big({\bf Q}^{\,(0)}\big)\big)
,\, w^{\,(0)}
\right)
\\
&=\left(\eta_{t}\big({\bf Q}^{\,(0)}\big) ,\, {\bf P}^{\,(0)}\cdot\nabla_{\,}\eta_t^{-1}\big({\bf Q}^{\,(0)}\big)-\nabla\big(w\,\beta_0\big({\bf Q}^{\,(0)}\big)\big)
,\, w^{\,(0)}
\right)
\\
&=\text{\large$\tau$}_{\!-\,{\rm d}(w\,\beta_0)\!}\circ \eta_t\ \Big({\bf Q}^{\,(0)},\, {\bf P}^{\,(0)},\, w^{\,(0)}\Big)
\,,
\end{align*}
\rem{ 
\comment{
CT: The cotangent lift action $T^*\eta(q,p)=(\eta(q),p\cdot T\eta^{-1}(q))$ is identical to that used in the plasma paper MaWeRaScSp83.

\bigskip
This notation is ambiguous. My definition in eqn (B.5.2) of GeomMech Part 2 page 240 has clearer notation. However, it would introduce double subscripts, so let's stay with this unless you think of something better. }
} 
where $(\eta_{t\,},\beta_0)\in\,{\rm Diff}\,\circledS\,\mathcal{F}$. It is
worth noticing that the order of operations in the composition is not relevant, since
\[
\eta^{-1}_t\circ\text{\large$\tau$}_{\!w\,{\rm d}\beta_0\!}\ \Big({\bf Q}^{\,(0)},\, {\bf P}^{\,(0)},\, w^{\,(0)}\Big)
=
(\eta_t,\beta_0)^{-1}\ \Big({\bf Q}^{\,(0)},\, {\bf P}^{\,(0)},\, w^{\,(0)}\Big)
\,,
\]
so that $(\eta^{-1},\beta_0)=(\eta,-\beta_0)^{\,-1}\in\,{\rm Diff}\,\circledS\,\mathcal{F}$. This can be easily seen from the following calculation, where we take $w=1$ for simplicity,
\begin{align*}
(\eta,\beta_0)^{-1}\circ\,(\eta,\beta_0)\cdot\left({\bf Q} ,\,{\bf P}\right)
&=
(\eta^{-1},-\beta_0)\cdot
\left(\eta_{t}({\bf Q}) ,\, {\bf P}\cdot^{\,} T_{\,}\eta_t^{-1}({\bf Q})+\textrm{d}\beta_0({\bf Q})
\right)
\\&
=
T^*\eta^{-1}\cdot
\left(\eta_{t}({\bf Q}) ,\, {\bf P}\cdot^{\,} T_{\,}\eta_t^{-1}({\bf Q})
\right)
\\&
=
T^*\eta^{-1}\circ\,T^*\eta\cdot({\bf Q,P}) =({\bf Q,P})
\,.
\end{align*}
Consequently, exchanging the order of operations in the composition simply yields another element of ${\rm Diff}\,\circledS\,\mathcal{F}$ and the arguments above are still valid.
Such group operations are useful in deriving fluid descriptions from kinetic equations, where the so called \emph{plasma-to-fluid momentum map} determines the Hamiltonian structure of the fluid system \cite{MaWeRaScSp}.

\begin{remark}[Geodesic motion on fiber-preserving transformations]
The arguments \, \,  above show that $EP\!\left({\rm Diff}_{\,}\circledS_{\,}\mathcal{F}\right)$ is a geodesic motion on the Lie group ${\rm Symp}_{\pi}(T^*{\rm Emb})$ of {\bfi fiber-preserving symplectomorphisms} on the cotangent bundle $T^*{\rm
Emb}$. It is a general result that any transformation by  ${\rm Symp}_{\pi}(T^*{\rm Emb})$ can be realized as the composition of a fiber translation and a cotangent lift \cite{BaWe97}. On the other
hand, such a transformation is always a canonical transformation characterized by a generating function that is \emph{linear} in the canonical momentum. Therefore, $EP\!\left({\rm Diff}_{\,}\circledS_{\,}\mathcal{F}\right)\simeq EP_{\ }{\rm Symp}_{\pi\,}$. 
Following this line of reasoning, one may also consider geodesic flows on the entire symplectic group ${\rm Symp}\left(T^*{\rm Emb})\right.$. These flows apply to the moment hierarchy in kinetic theory and have been treated in \cite{GiHoTr05,GiHoTr07}.
\end{remark}

Another important property of the momentum map $\bf J$ in (\ref{mommap-formula}) is its equivariance \cite{MaRa99}, which
guarantees that $\bf J$ is also a Poisson map. Below we show the infinitesimal equivariance of $\bf J$, which is defined by the following relation
\[
X_\beta[\big\langle{\bf J},\gamma \big\rangle]
= \big\langle {\bf J},{\rm ad}_\beta\, \gamma \big\rangle
\quad\forall\,\beta,\gamma\in\mathfrak{X}\,\circledS\,\mathcal{F}
\,.
\]
\begin{theorem}[Equivariance]
\label{equivar-thm}
The singular solution momentum map $\bf J$ in (\ref{mommap-formula})
\[
{\bf J}({\bf Q,P},w)\,=\!\int\!\big({\bf P}(s,t),w(s)\big)\,\delta\!\left({\bf
x-Q}(s,t)\right)\,{\rm d}^ks
\]
is infinitesimally equivariant.
\end{theorem}
\begin{proof}
\begin{align*}
X_\beta[\langle{\bf J},\gamma\rangle]
=&
X_\beta[\langle w,\gamma_0\rangle+\langle {\bf P},\boldsymbol{\gamma}_1\rangle]
\\
=&\int\left(
\left(w\,\frac{d \gamma_0}{d {\bf Q}}+{\bf P}\,\frac{d \boldsymbol{\gamma}_1}{d
{\bf Q}}\right)\cdot\boldsymbol{\beta}_1({\bf Q}(s))
-
\left(w\,\frac{d \beta_0}{d {\bf Q}}+{\bf P}\,\frac{d \boldsymbol{\beta}_1}{d
{\bf Q}}\right)\cdot
\boldsymbol{\gamma}_1
\right){\rm d}^ks
\\
=&\int
{\bf P}(s)\cdot\left(
\boldsymbol{\beta}_1\cdot\frac{d\boldsymbol{\gamma}_1}{d {\bf Q}}
-
\boldsymbol{\gamma}_1\cdot\frac{d \boldsymbol{\beta}_1}{d
{\bf Q}}\right){\rm d}^ks
\,+\!
\int\!
w(s)\left(
\boldsymbol{\beta}_1\cdot\frac{d\gamma_0}{d {\bf Q}}
-
\boldsymbol{\gamma}_1\cdot\frac{d\beta_0}{d
{\bf Q}}\right){\rm d}^ks
\\
=&
\int\left(\int
{\bf P}(s,t)\,\delta({\bf x-Q}(s,t))\,{\rm d}^ks\right)\cdot\bigg[
(\boldsymbol{\beta}_1\cdot\nabla)\boldsymbol{\gamma}_1
-
(\boldsymbol{\gamma}_1\cdot\nabla)\boldsymbol{\beta}_1
\bigg]\,{\rm d}^n{\bf x}
\\
&\,+\!
\int\left(
\int w(s)\,\delta({\bf x-Q}(s,t))\,{\rm d}^ks\right)\bigg[
(\boldsymbol{\beta}_1\cdot\nabla)\gamma_0
-
(\boldsymbol{\gamma}_1\cdot\nabla)\beta_0\bigg]\,{\rm d}^n {\bf x}
\\
=&\ \langle {\bf J},{\rm ad}_\beta\, \gamma\rangle
\,.
\end{align*}
\end{proof}


\subsection{Collective Hamiltonian}
From the expression of the vector field (\ref{HamVF}), we can immediately write the equations of motion for $\bf Q$ and $\bf P$. Moreover, if we insert the singular solution momentum map into the Hamiltonian (\ref{EPGosH-Ham}) we recover the following collective Hamiltonian $H_N: \times_{i=1}^N\left(T^*{\rm Emb}(S,\mathbb{R}^n)\times{\rm Den}(S)\right)\to\mathbb{R}$
\begin{multline*}
H_N\, =\, \frac12\sum_{i,j}^N\iint{\bf P}_i(s,t)\cdot{\bf P}_j(s',t)\ G_1\!\left({\bf Q}_{i}(s,t)-{\bf Q}_j(s',t)\right)\,{\rm d}^ks\,{\rm d}^ks'
\\
+
\frac12\sum_{i,j}^N\iint w_i(s)\,w_j(s')\, G_2\!\left({\bf Q}_i(s,t)-{\bf Q}_j(s',t)\right)\,{\rm d}^ks\,{\rm d}^ks'
\,.
\end{multline*}
The corresponding canonical Hamiltonian equations of motion are
\begin{align*}
\pp{{\bf Q}_i}{t}=&\sum_j\int {\bf P}_j(s',t)\ G_1\!\left({\bf Q}_i(s,t)-{\bf Q}_j(s',t)\right)\,{\rm d}^ks'
\,,\\
\pp{{\bf P}_i}{t}=&-
\,\sum_j\int \left({\bf P}_i(s,t)\cdot{\bf P}_j(s',t)\right) \pp{}{{\bf
Q}_i}G_1\!\left({\bf Q}_i(s,t)-{\bf Q}_j(s',t)\right)\,{\rm d}^ks'
\\
&-
\sum_j\int w_i(s,t)\,w_j(s',t)\, \pp{}{{\bf
Q}_i}G_2\!\left({\bf Q}_i(s,t)-{\bf Q}_j(s',t)\right)\,{\rm d}^ks'
\,.
\end{align*}

\begin{remark}[Singular potential terms and integrable cases]
In the limiting case $G_2=\delta$, we recover the integrable CH2 system corresponding to a geodesic motion with a delta-like potential. In terms of particle dynamics, this means that for a positive potential, two particles will bounce off immediately before colliding; while for a negative potential, they will proceed together, attached one to the other and they will never split apart. This kind of singular delta-like potential is also present in another integrable system, the Benney equation \cite{Be1973,Gi1981}.
\end{remark}

\subsection{Kaluza-Klein formulation}
It is interesting that the collective Hamiltonian allows for a Kaluza-Klein formulation similar to the usual treatment of a particle in a magnetic field \cite{MaRa99}. In order to see this, one first observes that the equations of motion for $({\bf Q}_i,{\bf P}_i,w_i)$ may be recovered on the Lagrangian side via the Legendre transform,
\[
{\bf V}^i=G_1^{\,ij}{\bf P}_j
\,,\qquad\qquad
\dot{\theta}^{\,i}=G_2^{\,ij}w_j
\,.
\]
This yields the following Lagrangian $L:\times_{i=1}^N\left(T{\rm Emb}(S,\mathbb{R}^n)\times\mathcal{F}(S)\right)\to\mathbb{R}$
\begin{multline*}
L_N\, =\, \frac12\sum_{i,j}^N\iint{\bf V}^i(s,t)\cdot{\bf V}^j(s',t)\ G^1\!\left({\bf Q}_{i}(s,t)-{\bf Q}_j(s',t)\right)\,{\rm d}^ks\,{\rm d}^ks'
\\
+
\frac12\sum_{i,j}^N\iint \dot{\theta}^i(s)\,\dot{\theta}^j(s')\, G^2\!\left({\bf Q}_i(s,t)-{\bf Q}_j(s',t)\right)\,{\rm d}^ks\,{\rm d}^ks'
\end{multline*}
where $G^i$ with raised index is the inverse metric associated to $G_i$. (If $G_i$ is given by a convolution kernel, then $G^i$ becomes a differential operator.) As our notation may suggest, we now enlarge our configuration space so that the Lagrangian $L_N$ becomes defined on
\[
TQ_{KK}:=T\left({\rm Emb}(S,\mathbb{R}^n)\times\mathcal{F}(S)\right)
=
T{\rm Emb}(S,\mathbb{R}^n)\times T\mathcal{F}(S)
\]
and $Q_{KK}$ is called {\bfi Kaluza-Klein configuration space}. Now, since
the coordinates $\theta_i$ in the Lagrangian $L_N=L({\bf Q}^i,{\bf V}^i,\theta^i,\dot{\theta}^i)$ are ignorable, its conjugate momenta $w_i$ will be constants of motion.
This allows the collective Hamiltonian to be naturally written in the Kaluza-Klein formulation on $T^*Q_{KK}$.  In this framework, it is well known \cite{RaTuSbSoTe05} that the weights $w_i$ are another type of conserved momentum map.
\begin{remark}[Physical interpretation]
The physical system described by the Kaluza-Klein Hamiltonian
$H_N$ turns out to be related to the motion of electrical charges whose mutual interaction is given by the potential term in $G_2$. This relation is evident by noticing that for the case
of a single particle, the Hamiltonian $H_N$ reduces to the Kaluza-Klein Hamiltonian
$H_1={\bf P}^2/2+w^2/2$ of a free charge \cite{MaRa99}. In the multi-particle
case, the momentum of a single charge is affected not only by the momenta
of the remaining particles in the system, but also by their charges involved
in the potential term.

The absence of an external magnetic field in the system is reflected the absence of its vector potential $\bf A$ in the Hamiltonian $H_N$. One may imagine inserting a non-zero magnetic field into the system by substituting
$\,{\bf P}_i(s)\to {\bf P}_i(s) - w_i(s){\bf A}({\bf Q}_i(s))$ into the Hamiltonian $H_N$. This is the usual process in electromagnetic (or Yang-Mills) theory.
\end{remark}

\subsection{A right action momentum map and the Kelvin-Noether theorem} As shown in the previous section, the singular solutions identify a momentum map which is determined
by a left action of the  group Diff$(\mathbb{R}^n)$  on ${\bf Q}\in{\rm Emb}(S,\mathbb{R}^n)$,
i.e. $\eta{\bf Q}=\eta\circ{\bf Q}$. However, as in the case of EPDiff, one may also construct a right action by ${\bf Q}\eta={\bf Q}\circ\eta$, which is defined through the group Diff($S$), rather than Diff($\mathbb{R}^n$).

In order to perform such a construction, we consider the Kaluza-Klein formulation from the previous section, so that the configuration space is now $Q_{KK}={\rm Emb}(S,\mathbb{R}^n)\times{\rm Den}(S)$. The group ${\rm Diff}(S)$ acts on $Q_{KK}$ from the right according to
\begin{align*}
\Big({\bf Q}^{\,(0)} ,\,\theta^{(0)}\Big)\eta_t\,&=
\left({\bf Q}^{\,(0)}\circ\eta_{t\,} ,\,\theta\circ\eta_t
\right)
\,.
\end{align*}
The cotangent lift of this action to $T^*Q_{KK}$ yields the
following.

\begin{theorem}
\label{rightactionmommap-thm}
The map
\begin{align}\label{right-action-mommap}
{\bf J}_{\rm S}({\bf Q,P},\theta,w)
\,
=&\ {\bf P}(s)\cdot{\rm d}{\bf Q}(s)+w(s){\rm d}\theta(s)
\end{align}
is a momentum map
\[
{\bf J}_{\rm S\,}:_{\,}T^{*\!}\big({\rm Emb}(S,\mathbb{R}^n)\times\mathcal{F}(S)\big)\to\,\mathfrak{X}^*(S)
\]
corresponding to the cotangent lift of the right action of ${\rm Diff}(S)$ on ${\rm Emb}(S,\mathbb{R}^n)\times\mathcal{F}(S)$. This quantity is preserved
by the flow generated by the Hamiltonian $H_N$.
\end{theorem}
\begin{proof}
Although we might proceed analogously to \cite{HoMa2004} by applying the formula for momentum maps arising from cotangent lifts, here we prefer to proceed by applying the general definition, since it exhibits the infinitesimal right action on each of the variables. Inserting the functional
\[
\Big\langle {\bf J}_{\rm S},\,\beta_1 \Big\rangle
=\int\!\big({\bf P}\cdot\nabla_{\!s}{\bf Q}\,+\,
w\,\nabla_{\!s}\theta\big)\cdot\beta_1(s)\, {\rm d}^k s
\]
into the Poisson bracket $\{F,\left\langle {\bf J},\beta_1\right\rangle\}$
yields
\begin{multline*}
\{F,\left\langle {\bf J}_{\rm S},\beta_1\right\rangle\}
=
\int\!\left(
{\beta}_1(s)\cdot
\nabla_{\!s}{\bf Q}\,
\cdot
\frac{\delta F}{\delta {\bf Q}}
\, +\,
\nabla_{\!s}\cdot(\beta_{1\,}{\bf P})
\cdot
\frac{\delta F}{\delta \bf P}
\right)
{\rm d}^ks
\\
+
\int\!\left(
{\beta}_1(s)\cdot
\nabla_{\!s}\theta
\
\frac{\delta F}{\delta \theta}
\, +\,
\nabla_{\!s}\cdot(w\,\beta_1)
\,
\frac{\delta F}{\delta w}
\right)
{\rm d}^ks
=\xi_{\beta_1}[F]
\end{multline*}
where $\xi_{\beta_1}$ is the infinitesimal generator corresponding
to cotangent lifts generated by $\beta_1\in\mathfrak{X}(S)$.

The last part of the statement follows from the fact that the Hamiltonian $H_N$ is invariant under the cotangent lift of the right action of Diff(S), which amounts to invariance of the integral over S under reparametrization using the change of variables formula.
\end{proof}

\begin{corollary}[Kelvin-Noether theorem]$\quad$\\
\label{cor.abelianKNthm}
Exactly as in the case of EPDiff, the conservation of the momentum map in Theorem \ref{rightactionmommap-thm} leads
naturally to the Kelvin-Noether theorem. Indeed, upon recalling the evolution of $w$, i.e. $\dot{w}+\nabla_{\!s}\cdot\left(w\,\beta_1\right)=0$,  one checks that
\[
\left(\frac{\partial}{\partial t}+\pounds_{\beta_1}\right)
\Big(
w^{-1}(s)\
{\bf P}(s)\cdot{\rm d}{\bf Q}(s)+{\rm d} \theta(s)
\Big)=0
\,,
\]
and therefore
\[
\frac{d}{dt}\oint_{\gamma_t\!} \big(w^{-1}(s)\
{\bf P}(s)\cdot{\rm d}{\bf Q}(s)+{\rm d}\theta\big)
=
\frac{d}{dt}\oint_{\gamma_t} w^{-1}(s)\
{\bf P}(s)\cdot{\rm d}{\bf Q}(s)
=
0
\,,
\]
where $\gamma_t$ is any Lagrangian circuit and the first step is justified by Stokes' theorem.
\end{corollary}

\begin{remark}[Dual pair structures]
Upon recalling from \cite{HoMa2004} that the term ${\bf J}_1:={\bf P}\cdot{\rm
d}{\bf Q}$ in (\ref{right-action-mommap}) is a momentum map ${\bf J}_1:T^*{\rm Emb}\to\mathfrak{X}^*(S)$ for the right action of Diff(S), one may construct the same dual pair structure as in the geometric description of the EPDiff equation. Indeed, we may introduce the map
\[
{\bf J}_{\rm Sing}({\bf Q},{\bf P})=\int\!{\bf P}(s,t)\,\delta({\bf x-Q}(s,t))\,{\rm d}^ks
\]
that is, the $\mathfrak{X}^*$-component of the singular solution momentum map ${\bf J}$ in Theorem \ref{singsolnmommap-thm}. This is well known \cite{HoMa2004}
to be a momentum map ${\bf J}_{\rm Sing}:T^*{\rm Emb}\to\mathfrak{X}^*(\mathbb{R}^n)$
that generates the left leg of the following dual pair picture\\
\begin{picture}(150,100)(-70,0)%
\put(100,75){$T^{\ast} {\rm Emb}(S,\mathbb{R}^n)$}

\put(78,50){$\mathbf{J}_{\rm Sing}$}

\put(160,50){$\mathbf{J}_1$}

\put(72,15){$\mathfrak{X}^{\ast} (\mathbb{R}^n)$}

\put(170,15){$\mathfrak{X}^{\ast}(S)$}

\put(130,70){\vector(-1, -1){40}}

\put(135,70){\vector(1,-1){40}}

\end{picture}\\
which is the standard dual pair picture associated to the EPDiff equation
\cite{HoMa2004}. Dual pair structures have been used in \cite{MaWe83} to explore the geometric nature of Clebsch variables in fluid systems. In the present context the variables $({\bf Q,P})$ form the Clebsch representation associated to the diffeomorphism group ${\rm Diff}(S)$. 
We refer to the original works \cite{MaWe83,We83} for deeper discussions on the geometric nature of dual pairs.

In the present case the right-action momentum map ${\bf J}_S$ also takes
into account the extra term ${\bf J}_2:=w\,{\rm d}\theta$, which is associated with the space of scalar functions $\mathcal{F}(S)$. This term is again a momentum
map ${\bf J}_2:T^*\mathcal{F}\to\mathfrak{X}^{\ast} (S)$, which
will be used later in this paper for the construction of
another dual pair, associated to the $\mathcal{F}^*$-component
of the singular solution momentum map in Theorem \ref{singsolnmommap-thm}.
\end{remark}


\bigskip
\section{Extension to anisotropic interactions}
\label{sec.3}

This section extends the previous results on singular solutions to the case when the fluid motion depends on an extra degree of freedom, such as the fluid particle orientation.
This occurs for example in the theory of liquid crystals \cite{Ho2002}. A geometric fluid theory for such systems is already present in the literature regarding Yang-Mills
charged fluids and quark-gluon plasmas \cite{GiHoKu1982,GiHoKu1983}. This work formulates the equations for the fluid momentum $\bf m(x)$, mass density $\rho({\bf x})$ and charge density $C({\bf x})$, where the charge is considered as an extra degree of freedom of each fluid particle. (This is the colour charge in the case of chromo-hydrodynamics for quark-gluon plasmas.) These equations are written
as
\begin{align}
\pp{\rho}{t}&+{\rm div}\left(\rho\,\dede{H}{\bf m}\right)=0
\nonumber
\,,\\
\pp{C}{t}&+{\rm div}\left(C\,\dede{H}{\bf m}\right)={\rm ad}^*_{\dede{H}{C}}\,C
\label{EPChromo}
\,,\\
\pp{\bf m}{t}&+\nabla\cdot\left(\dede{H}{\bf m}\otimes{\bf m}\right)+\left(\nabla\otimes\dede{H}{\bf m}\right)\cdot{\bf m}=-\,\rho\,\nabla\dede{H}{\rho}-\left\langle C,\nabla\dede{H}{C}\right\rangle
_{\mathfrak{g}^*\times\mathfrak{g}}
\,,
\nonumber
\end{align}
where $C$ takes values on the dual Lie algebra $\mathfrak{g}^*$, whose corresponding coadjoint operation is denoted by ad$^*$.
Thus the charge variable $C$ belongs to the space of $\mathfrak{g}^*$-valued
densities, which we denote by
\[
\mathfrak{g}^*(\mathbb{R}^n):={\rm Den}(\mathbb{R}^n)\otimes\mathfrak{g}^*
\,,
\]
so that $C\in\mathfrak{g}^*(\mathbb{R}^n)$. In what follows we shall use the elementary fact that the space $\mathfrak{g}^*(\mathbb{R}^n)$ is dual to $\mathfrak{g}(\mathbb{R}^n)=\mathcal{F}\otimes\mathfrak{g}$, and we shall use the same notation for $\mathfrak{g}$ and $\mathfrak{g}(\mathbb{R}^n)$. The distinction should be clear from the different contexts.

\rem{ 
\begin{remark}
At a careful analysis we should point
out that the right hand side of the second equation shows how the Lie algebra
$\mathfrak{g}$ is considered up to coadjoint orbits under its underlying
(gauge) Lie group $G$, so that $\mathfrak{g}$ should be rigourously substituted by
$\mathfrak{g}/G$. We shall emphasize this fact in the next section
where such a distinction is essential for our discussion.
\end{remark}
}    

Now, the equations above are known to possess a Lie-Poisson Hamiltonian structure dual to the Lie algebra of the semidirect-product Lie group ${\rm Diff}\,\circledS\left(\mathcal{F}\oplus\mathfrak{g}\right)$
\cite{GiHoKu1982,GiHoKu1983}.
One may also consider \emph{geodesic} Euler-Poincar\'e equations on this semidirect-product Lie group. This problem has already been considered in \cite{GiHoTr07} in terms of its singular solutions, although not in relation with momentum maps. For the sake of simplicity, we consider the semidirect product ${\rm Diff}\,\circledS\,\mathfrak{g}$ and denote the corresponding geodesic equations by EP(${\rm Diff}\,\circledS\,\mathfrak{g}$). (The commutative case $\mathfrak{g}=\mathcal{F}\times\mathbb{R}$ reduces to the case studied in Section \ref{sec.2}.) 

In order to construct
the EP(${\rm Diff}\,\circledS\,\mathfrak{g}$) equations, one writes a purely quadratic Hamiltonian
\begin{eqnarray}\label{chromoHam}
H({\bf m},C)
&=&
\frac12\iint\!{\bf m}({\bf x}) \cdot\, G_1({\bf x-x}')\,{\bf m}({\bf x}')\,{\rm d}^n{\bf x}\,{\rm d}^n{\bf x}'
\nonumber\\
&&\hspace{2cm}
+\
\frac12\iint\Big\langle C({\bf x}),\,G_2({\bf x-x}')\,C({\bf x}') \Big\rangle
_{\mathfrak{g}^*\times\mathfrak{g}}
\,{\rm d}^n{\bf x}\,{\rm d}^n{\bf x}'
\,,
\end{eqnarray}
which yields the geodesic equations in covariant form
\begin{align}\nonumber
C_t+\pounds_{G_1*\bf m}\ C&={\rm ad}^*_{\,G_2\,*\,C}\ C
\,,\\
\mathbf{m}_t+\pounds_{G_1*\bf m}\ {\bf m}&=C\diamond (G_2*C)
\,,
\label{EPDiffsigma}
\end{align}
where the Lie derivative and diamond operations were introduced in Section \ref{sec.1}. In order to simplify the discussion, one can specialize to the case when particles have an orientation (or spin) in space and think of the charge density as the distribution of the local particle orientation in space, so that $\boldsymbol{C}\in{\rm Den}(\mathbb{R}^n)\otimes\mathfrak{so}(3)\simeq{\rm Den}(\mathbb{R}^n)\otimes\mathbb{R}^3$ and the Lie bracket is given by the
usual cross product. However, the following result applies in general.

\begin{theorem}\label{singsolth-chromo}
The ${\rm EP}({\rm Diff}\,\circledS\,\mathfrak{g})$ equations admit singular solutions of the form
\begin{equation}
\big({\bf m},C\big)=\sum_{i=1}^N\int\!\big({\bf P}_i(s,t),\mu_i(s,t)\big)\,\delta\!\left({\bf
x-Q}_i(s,t)\right)\,{\rm d}^ks
\,,
\label{singsolnansatz}
\end{equation}
associated with the momentum map
\[
{\bf J}:\underset{i=1}{\overset{N}{\text{\LARGE$\times$}}}\big( T^* {\rm Emb}(S,\mathbb{R}^n)\,\times\,\mathfrak{g}^*(S)\big)\,
\to\ \mathfrak{X}^*(\mathbb{R}^n)\ \text{\large$\circledS$}\ \mathfrak{g}^*(\mathbb{R}^n)
\,.
\]
\end{theorem}
\begin{proof}
Again, we fix $i$ for convenience and suppress it in the notation. Substitution of the solution ansatz (\ref{singsolnansatz}) into the EP(Diff$\,\circledS\,\mathfrak{g}$) equations yields
\begin{equation}\label{singsol-eq}
\pp{{\bf Q}}{t}=\boldsymbol{\beta}_1({\bf Q})
\,,
\qquad\quad
\pp{{\bf P}}{t}=
-\,{\bf P}\,\cdot\text{\large$\nabla$}_{\bf \!Q\,}\boldsymbol{\beta}_1
-\big\langle \mu, \text{\large$\nabla$}_{\!\bf Q\,}\beta_0\big\rangle
\,,
\qquad\quad
\pp{\mu}{t}=
{\rm ad}^*_{\beta_0}\,\mu
\,.
\end{equation}
In order to define a momentum map, we first need to establish a Poisson structure on $T^*{\rm Emb}\,\bigoplus\,\mathfrak{g}^*$. To this purpose we use the following Poisson bracket \cite{GiHoKu1982,GiHoKu1983}:
\begin{equation}
\{F,G\}({\bf Q},{\bf P}, \mu)\,:=\int
\left(\dede{F}{\bf Q}\cdot\dede{G}{\bf P}-\dede{F}{\bf Q}\cdot\dede{G}{\bf P}\right)\,{\rm d}^k s\,
-
\int
\left\langle
\mu,\,\left[\dede{F}{\mu},\,\dede{G}{\mu}\right]
\right\rangle
\,{\rm d}^k s
\,,
\label{Poisson-br}
\end{equation}
where $\langle\cdot,\,\cdot\rangle$ denotes the pairing on $\mathfrak{g}^*\times\mathfrak{g}$.
Now, if $\beta=(\boldsymbol{\beta}_1,\beta_0)\in\mathfrak{X}(\mathbb{R}^n)\ \text{\large $\circledS$}\ \mathfrak{g}(\mathbb{R}^n)$, then the functional $\langle{\bf J},\beta\rangle$ may be defined as
\[
\left\langle{\bf J}({\bf Q},{\bf P}, \mu),\,\beta\right\rangle
=\int
\Big( \langle \mu,\,\beta_0({\bf Q})\rangle+{\bf P}\cdot\boldsymbol{\beta}_1({\bf
Q})\Big)\,{\rm
d}^ks
\,,
\]
where $\langle\cdot,\,\cdot\rangle$ may now denote the pairing on either $\mathfrak{g}^*\!\times\mathfrak{g}$
or $\mathfrak{X}^*\text{\large $\circledS$}\, \mathfrak{g}^{*\!} \times\mathfrak{X}\,\text{\large $\circledS$}\, \mathfrak{g}$. (No confusion should arise from this notation.)
At this point one calculates the Poisson bracket using equation (\ref{Poisson-br}) as, cf. equation (\ref{HamVecFld2}),
\begin{align}
\big\{
F,\langle{\bf J},\boldsymbol{\beta}\rangle
\big\}
\,=&
\int\!\left(
\frac{\delta F}{\delta \bf Q}\cdot\boldsymbol{\beta}_1({\bf Q})
-
\left(\left\langle \mu,\,\frac{d \beta_0}{d {\bf Q}}\right\rangle
+
\frac{d \boldsymbol{\beta}_1}{d{\bf Q}}^T \cdot {\bf P}
\right) \cdot
\frac{\delta F}{\delta {\bf P}}
\right){\rm d}^ks
+
\int\!\left\langle{\rm ad}^*_{\beta_0}\,\mu,\,\dede{F}{\mu}\right\rangle\,{\rm d}^ks
\,.
\label{HamVecFld3}
\end{align}
Thus, one finds the Hamiltonian vector field $\{F,\langle{\bf J},\boldsymbol{\beta}\rangle\}=X_\beta[F]$, with $({\bf Q},{\bf P}, \mu)$ components
\begin{align*}
X_\beta:=\bigg(\boldsymbol{\beta}_1({\bf Q}),\,
-\,
\left\langle \mu,\,\frac{d \beta_0}{d {\bf Q}}\right\rangle
-
\frac{d \boldsymbol{\beta}_1}{d{\bf Q}}^T \cdot {\bf P},
\,
{\rm ad}^*_{\beta_0}\,\mu
\bigg)
\,.
\end{align*}
The first two components of this vector field identify a Hamiltonian vector field on $T^*{\rm Emb}$ corresponding to the Hamiltonian
\[
H
=\int \Big( \langle \mu,\, \beta_0({\bf Q})\rangle\ +\ {\bf P}\cdot
\boldsymbol{\beta}_1({\bf Q})\Big)\, {\rm d}^ks=\langle{\bf J},\beta\rangle
\,,
\]
which generates compositions of cotangent lifts $T^*$Diff generated by $\boldsymbol{\beta}_1$ with fiber translations
\[
\text{\large$\tau$}_{\!-{\rm d}\langle\mu,\beta_0\rangle}\cdot (q,p) = (q,p-\langle\mu,{\rm d}\beta_0\rangle)
\]
generated by $-\langle\mu,\beta_0\rangle$. The third component generates pure coadjoint motion of the charge variable $\mu$  on $\mathfrak{g}^*$ according to
\[
\mu^{(t)}={\rm Ad}^*_{\,\exp\left(-\,t\,\beta_0\right)}\ \mu^{(0)}
\,,
\]
under the action of the Lie group $G$ whose underlying Lie algebra is $\mathfrak{g}:=T_e G$. Thus, the three-component vector field $X_\beta$ is an infinitesimal generator.
\end{proof}

\begin{remark}[Left action]
Just as in the case of $\rm EP(Diff\,\circledS\,\mathcal{F})$, we can write a similar
left group action of $\rm Diff\,\circledS\,\mathfrak{g}$ on $T^*{\rm Emb}\oplus\mathfrak{g}^*$.
Indeed, by the arguments in the proof above, one sees that the momentum map
$\bf J$ derives from the following left action
\begin{align*}
\Big({\bf Q}^{\,(t)} ,\,{\bf P}^{\,(t)},\, \mu^{\,(t)}\Big)&=
\left(\eta_{t\,}\circ{\bf Q}^{\,(0)} ,\, {\bf P}^{\,(0)}\cdot^{\,} T\big(\eta_t^{-1}\circ{\bf Q}^{\,(0)}\big)-\,{\rm d}\big(\big\langle\mu,\beta_0\big\rangle\circ{\bf Q}^{\,(0)}\big)
,\, {\rm Ad}^*_{\,\exp\left(-\,t\,\beta_0\right)}\ \mu^{(0)}
\right)
\,,
\end{align*}
where $\eta_t=\exp\left(t\,\boldsymbol{\beta}_1\right)\in{\rm Diff}$. The equivariance of $\bf J$ is proved via the same steps as we followed in the isotropic case.
\end{remark}

\begin{remark}[Principal bundle structure of the configuration space]
The proof above uses the fact that the space $T^*{\rm Emb}\times\mathfrak{g}^*$ may be endowed with a Poisson structure. The question remains on how to identify a suitable configuration manifold corresponding to such a phase space. The
answer can be found in the theory of gauged Lie-Poisson structures \cite{Mo1984,MoMaRa84}, such as the one we used in our proof. Indeed, it is well known how this structure arises from the reduced phase space of the configuration principal bundle $B$ 
\[
B={\rm Emb}(S,\mathbb{R}^n)\times G(S)
\,,
\]
where $G(S)$ is the gauge group underlying $\mathfrak{g}(S)$, consisting of $G-$valued maps $g:S\to G$, so that $T_e G(S)=\mathfrak{g}(S)$. In what follows
we shall use the same notation for $G$ and $G(S)$ and no confusion should
arise from this choice.

It is also known that the phase space corresponding to such a configuration manifold is obtained by reduction of the cotangent bundle $T^* B$. Thus our
phase space can be obtained as
\[
T^*{\rm Emb}\times\mathfrak{g}^*
\simeq
T^*B/G
\simeq
\big(T^*{\rm Emb}\times T^*G\big)/G
\simeq
T^*{\rm Emb}\times T^* G/G
\,,
\]
where the last step is justified by the fact that $G$ does not act on $\rm
Emb$. Therefore our geometric treatment follows from the principal bundle structure of the configuration space $B$.

In the Abelian case studied in Section \ref{sec.2}, one has $G=\mathcal{F}$ so that one may identify the gauge group with its Lie algebra. This identification is peculiar of Abelian gauge groups and cannot be performed in general. The next section shows how to treat  the Kaluza-Klein formulation of the non-Abelian case.
\end{remark}

\subsection{Kaluza-Klein collective Hamiltonian}
As in the isotropic case, we see again that the collective Hamiltonian
\begin{multline*}
H_N=
\frac12\,\sum_{i,j}^N
\iint\! {\bf P}_i(s)\cdot {\bf P}_j(s')\ G_1({\bf Q}_i(s)-{\bf Q}_j(s'))\, {\rm d}^ks\,{\rm d}^ks'
\\
+
\frac12\,\sum_{i,j}^N
\iint \!\Big\langle\mu_i(s),\,G_2({\bf Q}_i(s)-{\bf Q}_j(s'))\ \mu_j(s') \Big\rangle {\rm d}^ks\,{\rm d}^ks'
\,,
\end{multline*}
obtained by direct substitution of $\bf J(Q,P)$ in the EP(${\rm Diff}\,\circledS\,\mathfrak{g}$) Hamiltonian, allows for a Kaluza-Klein formulation. However, in this case
the gauge group is not Abelian and we need to proceed more carefully.
In the Kaluza-Klein picture of the motion of a colored particle in a Yang-Mills field, the particle motion is a geodesic on a principal $G$-bundle $B$. The metric on $B$ is $G$-invariant and its geodesics are determined by the $G$-invariant quadratic Hamiltonian on $T^*B$ where the Poisson bracket
is canonical. Specializing to our case yields
\[
Q_{KK}:=B=_{\!}\underset{i=1}{\overset{N}{\text{\LARGE$\times$}}}\big(\,{\rm Emb}\times G\,\big)
\]
and since the second term in the Hamiltonian $H_N$ is $G$-invariant by hypothesis, it  may be lifted to $T^*Q_{KK}$ as follows
\begin{multline*}
H_N({\bf Q}^i,{\bf
P}_i,g^i,p_i)=
\frac12\,\sum_{i,j}^N
\iint\! {\bf P}_i(s)\cdot {\bf P}_j(s')\ G_1({\bf Q}^i(s)-{\bf Q}^j(s'))\, {\rm d}^ks\,{\rm d}^ks'
\\
+
\frac12\,\sum_{i,j}^N
\iint \!\left\langle p_i(s),\,G_2({\bf Q}^i(s)-{\bf Q}^j(s'))\ p_j(s')\right\rangle {\rm d}^ks\,{\rm d}^ks'
\,,
\end{multline*}
where $p_i$ is the conjugate momentum of the group coordinate $g^i\in G$, so that $(g^i,p_i)\in T^*G$, and $\langle\cdot,\,\cdot\rangle$ is now the pairing between tangent and cotangent vectors on $G$. The  momentum $p_i$ is conserved, since it is conjugate to the cyclic variable $g^i$. Thus the Hamiltonian $H_N$ is Klauza-Klein and thereby recovers the conservation of $p_i$.
Such a conservation law becomes coadjoint motion on the dual Lie algebra $\mathfrak{g}^*$,
such that
\[
\dot\mu_i={\rm ad}^*_{\delta H_N/\delta \mu_i}\,\mu_i
\qquad\hbox{(no sum)}
\,,
\]
where $\mu_i\,=\,{g^{_{i\,}}}^{-1\,}\,p_i \ \ \, \forall i=1...N$ (no sum over $i$), exactly as happens for the motion of a rigid body, when $G=SO(3)$.
As a consequence of the above arguments, it is clear how the dynamics (\ref{singsol-eq}) of the singular solutions is Hamiltonian with respect to the Poisson bracket in (\ref{Poisson-br}), which is the sum of a canonical term and a Lie-Poisson term.
\rem{ 
\comment{How is $\mu_i$ related to $p_i$ and $g_i$? Is $\mu_i=g_i^{-1}p_i$ (for each $i$, that is, no sum)?
\\

What are the formulas for the singular solutions for this $J_{sing}$?
\\

What are the dynamical equations for the singular solutions? Are they canonical plus Lie-Poisson?
\\

What are the two-body scattering rules for head-on and overtaking collisions of solutions resulting from these equations?  

\bigskip
CT: The first three questions are now answered in the lines above.}
}  

\rem{ 
In order to see this, one redefines the configuration space in the usual way so that it becomes
the associated bundle
\[
Q_{KK}:=_{\!}\underset{i=1}{\overset{N}{\text{\huge$\times$}}}\big({\rm Emb}(S,\mathbb{R}^n)\times_G\mathfrak{g}(S)\big)
=_{\!}\underset{i=1}{\overset{N}{\text{\huge$\times$}}}\big({\rm Emb}(S,\mathbb{R}^n)\times\mathfrak{g}(S)/G\big)
\]
where $G$ does not act on ${\rm Emb}$  and quotient space $\mathfrak{g}/G$ consists of all the orbits of the type ${\rm Ad}_{g\,}\theta$. Now, since the Hamiltonian is a function on $T^*Q_{KK}$,  it depends on ${\bf
Q}_i$, ${\bf P}_i$ and the orbit spaces
\[
\left({\rm Orb}(\theta^i),{\rm Orb}(\mu^i)\right):=\left\{({\rm Ad}_{g^{-1\,}} \theta^i,{\rm Ad}^*_{g\,} \mu_i)\,|\,g\in G\right\}\subset T^*\mathfrak{g}(S)
\]
of $(\theta^i,\mu_i)\in T^*\mathfrak{g}(S)$, which will be denoted simply by $(\theta^i,\mu_i)$.
It is easy to see that, since $\theta^i\in\mathfrak{g}(S)/G$ are cyclic coordinates, then their conjugate momenta $\mu_i$ are constants of motion, so
that the space ${\rm Orb}(\mu_i)$ of the orbits of $\mu_i\in T^*\mathfrak{g}(S)$ will be unchanged under the Hamiltonian flow generated by $H_N$.
} 

\subsection{The right action momentum map and its implications}\label{sec-KK}
One may also consider the \emph{right} action through the group ${\rm Diff}(S)$. Upon following the same procedure as for the abelian case, one finds the momentum map corresponding to the right action of ${\rm Diff}(S)$
\begin{align}\label{right-momap}
{\bf J}_{\rm S}({\bf Q},{\bf P},g,p)={\bf P}(s)\cdot{\rm d}{\bf Q}(s)
+
\Big\langle p(s),{\rm d}g(s) \Big\rangle
\in\mathfrak{X}^*(S)
\end{align}
where  the Kaluza-Klein phase space is now
\[
T^*Q_{KK}=T^*{\rm Emb}\times T^*G
\,.
\]
This momentum map is again conserved because of the evident symmetry of the Hamiltonian $H_N$ under relabelling $s$ by a change of variables.

\begin{remark}[Kelvin-Noether theorem]$\quad$\\
Upon seeking a Kelvin-Noether theorem for the non-Abelian system, one recognises that this system does not provide any conserved density variable that could be used to construct the loop
integral of a differential one form, as done for EP$({\rm Diff}\,\circledS\,\mathcal{F})$ in Corollary \ref{cor.abelianKNthm}.
For this purpose, it suffices to fix a weight $w\in{\rm Den}(S)$ preserved by the flow to obtain the following circulation theorem
\begin{align*}
\frac{d}{dt}\oint_{\gamma_t}  w^{-1}(s)\, \big({\bf P}(s)\cdot{\rm d}{\bf Q}(s)+  \langle p(s),\,{\rm d}g(s)\rangle\big)
=0
\,.
\end{align*}
As in Section 2, it is interesting to notice that the momentum map ${\bf J}_{\rm S}$ for relabelling symmetry by right action is determined by the sum ${\bf J}_1+{\bf J}_2$ of two {\it distinct} momentum maps, one for {\rm Diff} and the other for the gauge symmetry,
\[
{\bf J}_1({\bf Q,P})={\bf P}(s)\cdot{\rm d}{\bf Q}(s)
\qquad\quad
{\bf J}_1:T^*{\rm Emb}\to\mathfrak{X}^{\ast} (S)
\,,
\]
\[
{\bf J}_2(g,p)= \big\langle p(s),{\rm d}g(s) \big\rangle
\qquad\quad
{\bf J}_2:T^*G\to\mathfrak{X}^{\ast} (S)
\,.
\]
These momentum maps have the same target space, but different image spaces. 
Now, since the pairing $\langle p(s),\,{\rm d}g(s)\rangle$
is invariant under the (left or right) $G$-action by cotangent lifts, it is possible to re-express it as
\[
\langle p(s),\,{\rm d}g(s)\rangle=
\langle g^{-1\,}p,\,g^{-1\,}\nabla_{\!s\,} g\rangle\, {\rm d}s
=:
\langle \mu(s), \boldsymbol{\mathcal{A}}(s)\rangle\, {\rm d}s
=
\langle \mu(s),\mathcal{A}(s)\rangle
\]
for a $\mathfrak{g}$-valued one form $\mathcal{A}(s)=\boldsymbol{\mathcal{A}}(s) {\rm d}s$ (i.e. a pure gauge connection). This result does not depend on the particular choice of left
or right $G$-action, since the invariance property is not affected by this
choice. Consequently, with the definitions $\mu:=g^{-1}p$ and $\mathcal{A}:=g^{-1}{\rm
d}g$ one may rewrite ${\bf J}_2$ as
\[
{\bf J}_2(g,p)= \big\langle p(s),{\rm d}g(s) \big\rangle
=\big\langle \mu(s),\mathcal{A}(s)\big\rangle
\,.
\]
\end{remark}

\begin{remark}
The quantity $\mathcal{A}$ determines the following magnetic component of a Yang-Mills field 
\[
{\cal B}={\rm d}^{\cal A}{\cal A}={\rm d}{\cal A}+\left[{\cal A},\,{\cal A}\right],
\]
which is localized on the embedded subspace $S$ and is intrinsically generated by the moving charge $\mu$. In fact, the connection $\mathcal{A}$
does not represent the Yang-Mills potential of an external force field (not present in this case), usually denoted by ${A}={\bf A}({\bf x})\cdot{\rm d}{\bf x}$, which is rather a $\mathfrak{g}$-valued one form over the physical space $\mathbb{R}^n$.
\newline
More formally, $\cal B$ is the curvature of the connection induced by 
$\mathcal{A}$, which becomes important also in the geometric approach to the dynamics of complex fluids \cite{GaRa08,Ho2002}. In such
an approach, the interpretation of $\mathcal{A}:=g^{-1}{\rm d}g$ as a group one-cocycle plays
a central role, as shown in  \cite{GaRa08}.
\end{remark}

The momentum map ${\bf J}_2$ can now
be used to construct another dual pair, describing the geometry of the dynamics of the gauge charge $\mu$. Indeed, it is well known \cite{MaRa99} that the expression
$\mu=g^{-1}p=:{\bf J}_R(g,p)$ is a momentum  map ${\bf J}_R:T^*G\to\mathfrak{g}^*$ associated to cotangent lifts of right translation.  One may use this map to construct the following dual pair picture

\begin{picture}(150,100)(-70,0)%
\put(118,75){$T^*G(S)$}

\put(92,50){$\mathbf{J}_{R}$}

\put(160,50){$\mathbf{J}_2$}

\put(72,15){$\mathfrak{g}^{\ast} (S)$}

\put(170,15){$\mathfrak{X}^{\ast}(S)$}

\put(130,70){\vector(-1, -1){40}}

\put(135,70){\vector(1,-1){40}}

\end{picture}\\
According to the general definition \cite{We83}, a pair of momentum maps $\mathfrak{h}^*\overset{J_1}{\longleftarrow}
\,P\overset{J_2}{\longrightarrow}\,\mathfrak{g}^*$ is called a {\it dual
pair} if and only if ${\rm Ker\,}TJ_1$ and ${\rm Ker\,}TJ_2$ are symplectically orthogonal to one another. As explained in \cite{HoMa2004}, a necessary condition for $\mathfrak{h}^*\overset{J_1}{\longleftarrow}
\,P\overset{J_2}{\longrightarrow}\,\mathfrak{g}^*$ to be a dual pair is that each Lie group $G_i$ associated to $J_i$ acts transitively on the level sets of $J_k$ with $k\neq i$. Now, Diff($S$) acts transitively on
the level sets of ${\bf J}_R=\mu(s)$, because of the parameterization
freedom. Moreover, the action of $G(S)$ on level sets of ${\bf J}_2$ is transitive too, since it is given by cotangent lifts.
Thus, similar arguments to those in \cite{HoMa2004} allow one to conclude that the above dual pair is properly defined. 


\section{Kaluza-Klein equations for semi-direct products}
As we have seen from the previous sections, the Kaluza-Klein construction explains how collective motion on semidirect-product Lie groups arises under the singular solution momentum map. In this section, we extend the Kaluza-Klein formulation to continuum equations on a semidirect-product Lie group. The resulting equations apply to the method of \emph{metamorphosis} in the problem of matching shapes using active templates in imaging science \cite{HoTrYo2007}. In this application, the zero level set of the momentum map for right action plays a crucial role.

\subsection{Case of Abelian gauge groups}
From the theory of semidirect-product reduction \cite{MaRaWe84a,MaRaWe84b,HoMaRa}, one knows that the reduced  Lagrangian $L:\mathfrak{X}\,\circledS\,\mathcal{F}\to\mathbb{R}$ is an invariant function on ${\rm Diff}\,\circledS\,\mathcal{F}$. The reduction process may be performed in two ways. The first is to reduce according to the (right) action of the whole group
\[
T\big({\rm Diff}\,\circledS\,\mathcal{F}\big)/_{\,}{\rm Diff}\,\circledS\,\mathcal{F}
\simeq
\mathfrak{X}\,\circledS\,\mathcal{F}
\,,
\]
while the second approach uses the fact that
\[
\Big(T\big({\rm Diff}\,\circledS\,\mathcal{F}\big)/\mathcal{F}\Big)/_{\,}{\rm Diff}
\simeq
\mathfrak{X}\,\circledS\,\mathcal{F}
\,,
\]
which is known as {\it reduction by stages} \cite{CeMaRa01}.

In order to construct the Kaluza-Klein Lagrangian for continuum semidirect-product motion, one must \emph{enlarge} the semidirect-product structure to incorporate the cyclic variables. Indeed, one can consider the invariant reduced Lagrangian as a function
\[
L_{KK}:\mathfrak{X}\times T\mathcal{F}\to\mathbb{R}
\,,
\]
where the Lie algebra $\mathfrak{X}$ acts on $T\mathcal{F}$ by the Lie derivative arising from the tangent lift of the (right) ${\rm Diff}$-action on $\mathcal{F}$.
\rem{ 
whose underlying Lie group is ${\rm Diff}\,\circledS\,T{\cal F}$.
Here the (right) ${\rm Diff}$-action on $T\mathcal{F}$ is constructed
by tangent lifts of the pullback. This yields the pullback on $T\mathcal{F}$,
so that $\eta (\theta,\dot\theta)=(\theta\circ \eta, \dot\theta\circ \eta)\,\ \forall
\eta\in {\rm Diff}$.
}    
In particular, we first consider the unreduced Lagrangian ${\cal L}(\eta,\dot\eta,f,\dot{f})$
on $T{\rm Diff}\times T\mathcal{F}$. Since this is Diff-invariant, one may construct the Kaluza-Klein Lagrangian $L_{KK}$ as
\[
{\cal L}(\eta,\dot\eta,f,\dot{f})=L_{KK}({\bf u},f\,\eta^{-1},\dot{f}\,\eta^{-1})=L_{KK}({\bf u},\varphi,\lambda)
\,,
\]
which then has exactly the same expression as the usual Lagrangian $L$ on $\mathfrak{X}\,\circledS\,\mathcal{F}$, cf. \cite{HoTrYo2007}
\[
L_{KK}({\bf u},\varphi,\lambda)=\frac12\int \!{\bf u} \ Q_1 {\bf u}\,\,{\rm d}^n{\bf x}+\frac12\int \!\lambda\
 Q_2\,\lambda\ {\rm d}^n{\bf x}=L({\bf u},\lambda)
 \,,
\]
although the function $\varphi$ is now considered as a cyclic variable, whose corresponding
momentum (the density $\rho$) is preserved by the flow.

\rem{ 
Upon Legendre transforming,
one can extend the treatment in the previous sections and construct the following
right action momentum map
\[
{\bf J}({\bf m},\theta,\rho)={\bf m}({\bf x})+\rho({\bf x})\,{\rm d}\theta({\bf x})
\]
such that
\[
{\bf J}:\mathfrak{X}^*(\mathbb{R}^n)\,\circledS\,T^*\mathcal{F}(\mathbb{R}^n)\to\mathfrak{X}^*(\mathbb{R}^n)
\]
The motion on $T^*\mathcal{F}$ is now given by the sum of canonical Hamiltonian
dynamics and an advection term, which is given by the action of diffeomorphisms
(relabeling symmetry). Therefore this momentum map is {\it not} conserved
a priori by the flow, since

It is straightforward to write the Kelvin circulation theorem in the form
\[
\frac{d}{dt}\oint_{\gamma_t} \frac1\rho\left({\bf m}+\rho\,{\rm d}\theta\right)
\,=\,
\frac{d}{dt}\oint_{\gamma_t} \frac{\bf m}\rho
\,=\,0
\]
thereby reproducing at the continuum level the results already obtained for collective motion.
}   

From the arguments above we recognize that the Kaluza-Klein Lagrangian is related to the Lagrangian for EP$({\rm Diff}\,\circledS\,\mathcal{F})$. Thus, it is natural to ask how this relation arises from a reduction process. One first observes that
\[
T\big({\rm Diff}\,\circledS\,\mathcal{F}\big)/{\rm
Diff}
\simeq
\big(T{\rm Diff}\times T\mathcal{F}\big)/{\rm
Diff}
\simeq
\mathfrak{X}\times\!\big(T\mathcal{F}/\,{\rm Diff}\big)
\,.
\]
The parenthesis in the last step indicates that functions in $T\cal
F$ are defined modulo the action of diffeomorphisms, so that any element $(n,\nu)\in T{\cal F}/_{\,}\rm Diff$ can be written as $(f,\dot{f})\,\eta^{-1}\,\forall \eta\in\rm Diff$.
This application of the Kaluza-Klein reduction process to the {\it unreduced} Lagrangian ${\cal L}$ differs from ordinary semidirect-product reduction, which proceeds through the following steps
\[
\Big(T\big({\rm Diff}\,\circledS\,\mathcal{F}\big)/\mathcal{F}\Big)/_{\,}{\rm Diff}
\simeq
\Big(T{\rm Diff}\times T\mathcal{F}/\mathcal{F}\Big)/_{\,}{\rm Diff}
\simeq
\mathfrak{X}\times\!\big(\mathcal{F}/_{\,}{\rm Diff}\big)
\simeq
\mathfrak{X}\,\circledS\,\mathcal{F}
\,.
\]
In conclusion, the two Lagrangians $L$ and $L_{KK}$ may be derived from the same configuration Lie group ${\rm Diff}\,\circledS\,\mathcal{F}$, by following {\it different} reduction processes.

One may imagine that the same process may be
followed upon replacing $\mathcal{F}$ by a generic vector space $V$, yielding the Lie algebra $\mathfrak{X}\,\circledS\,TV$. However, in such a case, the gauge Lie group needed to construct a Kaluza-Klein formulation is absent.

\subsection{Extension to non-abelian gauge groups}

All the considerations in the previous section
may also be carried out for the case of anisotropic interactions, simply  by replacing $\mathcal{F}$ by the non-Abelian gauge group $G$. The reduction process underlying the continuum dynamics on $\mathfrak{X}\,\circledS\,\mathfrak{g}$ proceeds as follows
\[
\Big(T\big({\rm Diff}\,\circledS\,G\big)/\mathcal{F}\Big)/_{\,}{\rm Diff}
\simeq
\Big(T{\rm Diff}\,\times\,TG/G\Big)/_{\,}{\rm Diff}
\simeq
\mathfrak{X}\,\times\,\mathfrak{g}/{\rm Diff}
\simeq
\mathfrak{X}\,\circledS\,\mathfrak{g}
\,.
\]
However, for $G$-invariant Lagrangians on $\mathfrak{X}\,\circledS\,\mathfrak{g}$,
we may interpret the dynamics as occurring on the space $\mathfrak{X}\times
TG$, so that the Kaluza-Klein Lagrangian
\[
L_{KK}:\mathfrak{X}\times TG\to\mathbb{R}
\]
with $(g,\dot{g})\in TG$ and $\eta\in {\rm Diff}$ is written after right reduction by Diff as
\[
{\cal L}(\dot{\eta}\eta^{-1},g\eta^{-1},\dot{g}\eta^{-1})=
L_{KK}({\bf u},n,\nu)=\frac12\int \!{\bf u} \ Q_1 {\bf u}\,\,{\rm d}^n{\bf x}+\frac12\int \!\big\langle Q_2\,\nu,\,\nu \big\rangle \ {\rm d}^n{\bf x}
\,,
\]
with definitions ${\bf u}:=\dot{\eta}\eta^{-1}$, $n:=g\eta^{-1}$ and $\nu=\dot{g}\eta^{-1}$.
Now, since such a Lagrangian is also $G$-invariant, then one may write
\[
L_{KK}({\bf u},n^{-1\,}n,n^{-1\,}\nu)=L({\bf u},\chi)
\,,
\]
where $n=g\eta^{-1\!}\in G$ as before and $L({\bf u},\chi)$ with $\chi=n^{-1\,}\nu$ is the Lagrangian on $\mathfrak{X}\,\circledS\,\mathfrak{g}$. Legendre transforming this Lagrangian produces the Hamiltonian $H({\bf m},\boldsymbol{C})$ in (\ref{chromoHam}).
This construction yields the conservation of the conjugate variable \makebox{$p=\delta L_{KK}/\delta\nu$} along the flow of the group of diffeomorphisms, since $n$ is an ignorable coordinate. The reduction process involved in such a system proceeds as follows
\[
T\big({\rm Diff}\,\circledS \,G\big)/{\rm
Diff}
\,\simeq\,
\big(T{\rm Diff}/{\rm Diff}\big)\!\times \!\big(TG/{\rm Diff}\big)
\,\simeq\,
\mathfrak{X}\times\!\big(TG/{\rm Diff}\big)
\,,
\]
where $TG$ is the group of tangent lifts of $G$, which is itself acted on by the diffeomorphisms. Thus, again the two Lagrangians $L$ and $L_{KK}$ may be derived from the {\it same} {\it unreduced Lagrangian ${\cal L}$}. 
Consequently, the geodesic motion on semidirect-product Lie groups of the kind ${\rm Diff}\,\circledS\,G$ {\it always} possesses a Kaluza-Klein construction.

\subsection{Application to metamorphosis}
Lagrangian formulations on ${\rm Diff}\,\circledS\,G$ have been recently considered in \cite{HoTrYo2007}, where the whole theory is extensively studied in the context of imaging science. The Euler-Poincar\'e equations corresponding to a Lagrangian $L({\bf u},n,\nu)$ carrying the cyclic variable $n=g\eta^{-1\!}\in G$ are found to be
\begin{align}\nonumber
\bigg(\frac{\pa}{\pa t}+\pounds_{\bf u\,}\bigg)
&\frac{\delta L}{\delta \bf u}
=
-\left\langle\frac{\delta L}{\delta \nu},{\rm d}\nu\right\rangle
\,,\\
\label{metamorphosis}
\bigg(\frac{\pa}{\pa t}+\pounds_{\bf u\,}\bigg)
&\frac{\delta L}{\delta \nu}
=
0
\,,\\
\nonumber 
\bigg(\frac{\pa}{\pa t}+\pounds_{\bf u\,}\bigg)& n =\,\nu \,,
\end{align}
in which the last equation arises from the partial time derivative of the definition $n=g\eta^{-1\!}\in G$.
These equations imply that the Legendre-transformed variable $p={\delta L}/{\delta \nu}$ is preserved by the flow, which does {\it not} occur in the general case, when $L$ depends also on $n$.

In the general case, one may obtain the dynamics directly from a constrained variational principle $\delta S=0$ with 
\[
S = \int \bigg[ 
L({\bf u},n,\nu) 
+ \bigg\langle 
p, \frac{\partial n}{\partial t} + \pounds_{\bf u\,}n - \nu
\bigg\rangle
\bigg]dt
\,,
\]
where the angle bracket denote $L^2$ pairing. Stationary variations produce
\begin{eqnarray}
0 =\delta S &=&
 \int \bigg[ 
 \bigg\langle 
\frac{\delta L}{\delta \bf u} - p \diamond n,\,{\delta \bf u} 
\bigg\rangle
+
\bigg\langle 
\delta  p, \frac{\partial n}{\partial t} + \pounds_{\bf u\,}n - \nu
\bigg\rangle
\nonumber
\\&&\quad
+\
\bigg\langle 
\frac{\delta L}{\delta \nu} - p,\,\delta  \nu
\bigg\rangle 
+
\bigg\langle 
\frac{\delta L}{\delta n} -  \frac{\partial p}{\partial t} 
+ \pounds_{\bf u\,}^\dagger p 
,\,\delta  n
\bigg\rangle
\bigg]dt
\,,
\end{eqnarray}
in which the diamond operator $(\,\diamond\,)$ is defined via the natural generalization of (\ref{diamond.def}) and $\pounds_{\bf u\,}^\dagger$ with superscript dagger  denotes the $L^2$ adjoint of the Lie derivative so that, in particular, $\langle\pounds_{\bf u\,}^\dagger p,\,\delta n\rangle = \langle p,\,\pounds_{\bf u\,}\delta n\rangle$. Standard Euler-Poincar\'e theory for the case that $n$ is a scalar function and its dual $p$ is a density then implies the following system after a brief calculation, 
\begin{align}\nonumber
\bigg(\frac{\pa}{\pa t}+\pounds_{\bf u\,}\bigg)
&
\bigg(\frac{\delta L}{\delta \bf u} - p \diamond n \bigg)
=0 
\,,\\
\label{metamorphosis1}
\bigg(\frac{\pa}{\pa t}+\pounds_{\bf u\,}\bigg)&
\frac{\delta L}{\delta \nu}
=
\frac{\delta L}{\delta n}
\,,\qquad
\frac{\delta L}{\delta \nu} = p
\,,\\
\nonumber 
\bigg(\frac{\pa}{\pa t}+\pounds_{\bf u\,}\bigg)&n\,=\,\nu 
\,.
\end{align}
System (\ref{metamorphosis1}) possesses an exchange symmetry between  the variables $(n,\,\nu)\in TG$ and their dual variables $(\delta L/\delta n,\,\delta L/\delta \nu)\in TG^*$, and it satisfies the following proposition.

\begin{proposition}\label{meta.equiv}
System (\ref{metamorphosis1}) is equivalent to system (\ref{metamorphosis}) when $\delta L/\delta n=0$.
\end{proposition}

\begin{proof}
This proposition follows from a direction calculation using the chain rule for the diamond operation and substituting the last three equations in system (\ref{metamorphosis1}). Namely,
\begin{align}
0
=
\bigg(\frac{\pa}{\pa t}+\pounds_{\bf u\,}\bigg)
&
\bigg(\frac{\delta L}{\delta \bf u} - p \diamond n \bigg)
\nonumber\\
=
\bigg(\frac{\pa}{\pa t}+\pounds_{\bf u\,}\bigg)
&
\frac{\delta L}{\delta \bf u}
- 
\frac{\delta L}{\delta n}\diamond n
-
p\diamond \nu
\,.
\nonumber
\end{align}
When $n$ is a scalar and $p$ is a density, then $-\,p\diamond \nu
= \langle p,\,d\nu \rangle$ and setting $\delta L/\delta n=0$ recovers the first equation in the system (\ref{metamorphosis}). 
\end{proof}

\begin{remark}
The Legendre transformation of the constrained Lagrangian 
defines the Hamiltonian
\begin{equation}
H({\bf m}, p,\,n) = 
\langle {\bf m},\,{\bf u}  \rangle
+
\langle p,\,\nu\rangle
- 
L({\bf u}, n, \nu)
\,,
\label{Legendre-trans}
\end{equation}
in terms of the fiber derivatives
\begin{equation}
{\bf m}:=\frac{\delta L}{\delta{\bf u} }
\quad\hbox{and}\quad
p:= \frac{\delta L}{\delta\nu}
\,.
\label{fiber.deriv}
\end{equation}
The variational derivatives of the Hamiltonian are found by substituting into the Legendre transformation (\ref{Legendre-trans}) as
\begin{eqnarray}
\delta H &=& 
\bigg\langle \frac{\delta H}{\delta{\bf m} },\,\delta{\bf m} \bigg\rangle
+
\bigg\langle \frac{\delta H}{\delta p},\,\delta{p} \bigg\rangle
+
\bigg\langle \frac{\delta H}{\delta{n} },\,\delta{n} \bigg\rangle
\nonumber\\
&=&
\bigg\langle {\bf u},\,\delta{\bf m} \bigg\rangle
+
\bigg\langle \nu,\,\delta{p} \bigg\rangle
-
\bigg\langle \frac{\delta L}{\delta{n} },\,\delta{n} \bigg\rangle
\nonumber\\&&
+\
\bigg\langle  {\bf m}-\frac{\delta L}{\delta{\bf u} },\,\delta{\bf u} \bigg\rangle
+
\bigg\langle p-\frac{\delta L}{\delta{\nu} },\,\delta{\nu} \bigg\rangle
\,.
\label{fiber.deriv}
\end{eqnarray}
Hence, $\delta H/\delta  {\bf m} = {\bf u}$, $\delta H/\delta p = \nu$ and the semidirect-product Lie-Poisson Hamiltonian equations corresponding to the system (\ref{metamorphosis}) are
\begin{align}\nonumber
&\bigg(\frac{\pa}{\pa t}+\pounds_{\delta H/\delta {\bf m}}\bigg){\bf m} 
=
-\, \frac{\delta H}{\delta n} \diamond n 
+ p\diamond \frac{\delta H}{\delta p}
\,,\\
\label{metamorphosis2}
&\bigg(\frac{\pa}{\pa t}+\pounds_{\delta H/\delta {\bf m}}\bigg)p
=
-\,\frac{\delta H}{\delta n}
\,,\\
\nonumber 
&\bigg(\frac{\pa}{\pa t}+\pounds_{\delta H/\delta {\bf m}}\bigg)n
 = \frac{\delta H}{\delta p} 
\,.
\end{align}
Perhaps not unexpectedly, the first equation may also be written to agree with the motion equation in (\ref{metamorphosis1}) as 
\begin{equation}
\bigg(\frac{\pa}{\pa t}+\pounds_{\delta H/\delta {\bf m}}\bigg){\bf m} 
=
\bigg(\frac{\pa}{\pa t}+\pounds_{\delta H/\delta {\bf m}}\bigg)
(p \diamond n)
\,.
\label{mom.conserv}
\end{equation}
\end{remark}

\subsection{The Kelvin circulation theorem}
Proposition \ref{meta.equiv} allows the Kelvin circulation theorem for these semidirect-product systems with $\delta L/\delta n=0$ to be expressed in two ways, upon introducing a conserved density variable $\rho$, satisfying
\[
\bigg(\frac{\pa}{\pa t}+\pounds_{\bf u\,}\bigg)\rho = 0
\,.
\]
On one hand, the first equation in the system (\ref{metamorphosis}) and the general theory of dynamics on semidirect-product Lie groups \cite{HoMaRa} imply, with definitions (\ref{fiber.deriv}), that
\begin{equation}
\frac{d}{dt}\oint_{\gamma_t} \frac{\bf m}\rho \,=\,-
\oint_{\gamma_t} \frac1\rho\, \big\langle p,{\rm d}\nu\big\rangle
\,,
\label{SDP.circ}
\end{equation}
where $\gamma_t$ is a closed loop moving with the flow of the velocity vector field. 
On the other hand, the first equation in the system (\ref{metamorphosis1}) implies that the Kelvin circulation theorem may also be expressed as
\begin{equation}
\frac{d}{dt}\oint_{\gamma_t} \frac1\rho\,
\Big({\bf m}+\big\langle p,{\rm d}n\big\rangle\Big)
\,=\,
0
\,,
\label{EPmom.circ}
\end{equation}
where the sum ${\bf m}_{tot}:={\bf m}+\big\langle p,{\rm d}n\big\rangle$ is the total momentum. This circulation theorem for total momentum is the natural extension to the continuum description of formula (\ref{right-momap}) for preservation of the right-invariant momentum map.
The two circulation laws (\ref{SDP.circ}) and (\ref{EPmom.circ}) are shown to be equivalent in the following.
\begin{proposition} The two forms of the Kelvin circulation theorem in equations (\ref{SDP.circ}) and (\ref{EPmom.circ}) are equivalent. 
\end{proposition}

\begin{proof}
This statement will hold, provided
\[
\frac{d}{dt}\oint_{\gamma_t} \frac1\rho\,\big\langle p,{\rm
d}n\big\rangle \,=\, \oint_{\gamma_t} \frac1\rho\, \big\langle
p,{\rm d}\nu\big\rangle \,,
\]
which may be verified directly, by using the second two equations in the system (\ref{metamorphosis}) as
\begin{align*}
\frac{d}{dt}\oint_{\gamma_t} \frac1\rho\,\big\langle p,{\rm d}n\big\rangle
&\,=\,
\oint_{\gamma_0} \frac1{\rho_0}\,\frac{d}{dt}\Big(\eta^*_t\big\langle p,{\rm d}n\big\rangle\Big)
\,=\,
\oint_{\gamma_0} \frac1{\rho_0}\,\eta^*_t\!\left(\frac{\partial}{\partial t}\big\langle p,{\rm d}n\big\rangle+\pounds_{\bf u\,}\big\langle p,{\rm d}n\big\rangle\right)
\\
&\,=\, \oint_{\gamma_0} \frac1{\rho_0}\,\eta^*_t\left( \left\langle
\frac{\partial p}{\partial t}+\pounds_{\bf u\,} p,{\rm
d}n\right\rangle + \left\langle p,{\rm d}\frac{\partial n}{\partial
t}+\pounds_{\bf u\,}{\rm d}n\right\rangle   \right)
\\
&\,=\, \oint_{\gamma_0} \frac1{\rho_0}\,\eta^*_t\left( \left\langle
p,{\rm d}\frac{\partial n}{\partial t}+ {\rm d}\pounds_{\bf
u\,}n\right\rangle \right)
\\
&\,=\, \oint_{\gamma_0} \frac1{\rho_0}\,\eta^*_{t\,} \big\langle
p,{\rm d}\nu\big\rangle \,=\, \oint_{\gamma_t} \frac1{\rho_t}\,
\big\langle p,{\rm d}\nu\big\rangle \,,
\end{align*}
where one recalls that the exterior differential commutes with the Lie derivative and the last step follows from the equation for $\partial n/\partial t$.
\end{proof}

\medskip
\noindent
This proposition extends the arguments in Section \ref{sec-KK} to the non-Abelian case in the continuum fluid description.

\begin{remark}
The zero level set of the total momentum, cf. equation (\ref{right-momap}),
\begin{equation}
{\bf m}+\big\langle p,{\rm d}n\big\rangle = 0
\,,
\label{zeromomtot}
\end{equation}
is preserved by the first equation in the system (\ref{metamorphosis}). The preservation of zero total momentum is a key step in the metamorphosis approach using active templates in imaging science, because the zero value is imposed by the requirement that an initial image would evolve to match a prescribed final image at a certain end point in time  \cite{HoTrYo2007}. 

The zero level set condition (\ref{zeromomtot}) for total momentum imposes the relation
\begin{equation}
{\bf m} = -\, \big\langle p,{\rm d}n\big\rangle
\,.
\end{equation}
This is the equivariant momentum map obtained from the cotangent-lift of the right action of Diff on the gauge group $G$, defined by 
\begin{equation}
\big\langle 
{\bf m},\, {\bf u} 
\big\rangle_{\mathfrak{g}^*\times\mathfrak{g}}
= \big\langle p\diamond n,\,{\bf u}\big\rangle
_{\mathfrak{g}^*\times\mathfrak{g}}
= -\,\big\langle p,\,\pounds_ {\bf u}n\big\rangle_{T^*G}
\,.
\end{equation}
Thus, the zero level set condition (\ref{zeromomtot}) for total momentum is itself a momentum map. This particular momentum map also appears in the application of the classical Clebsch method of introducing canonical variables for fluid dynamics. See, e.g., \cite{HoKu1983-Clebsch}. 
\end{remark}

\rem{ 
To proceed further one constructs the momentum map
\[
{\bf J}({\bf m},g,p)={\bf m(x)}+\langle p({\bf x}),\,{\rm d}g({\bf x})\rangle
\qquad\quad
{\bf J}:\mathfrak{X}^*\,\circledS\,T^*G\to\mathfrak{X}
\]
which is conserved by the Hamiltonian $H({\bf m},\boldsymbol{C})$ in (\ref{chromoHam}).
One then reduces
\[
\langle p,\,{\rm d}g\rangle=\langle g^{-1}p,\,g^{-1}{\rm d}g\rangle=\langle \boldsymbol{C},\,{\bf A}\rangle
\]
to construct
\[
{\bf J}^{A}({\bf m},\boldsymbol{C})={\bf m(x)}+\langle \boldsymbol{C}({\bf x}),\,{\bf
A}({\bf x})\rangle
\qquad\quad
{\bf J}^{A}:\mathfrak{X}^*\,\circledS\,\mathfrak{g}^*\to\mathfrak{X}
\]
where $\mathfrak{X}^*\,\circledS\,\mathfrak{g}^*$ is endowed with the following
Poisson bracket
\begin{multline*}
\{G,H\}=
-
\int
{\bf m}\cdot\left[
\left(\frac{\delta G}{\delta {\bf m}}\,
\cdot\nabla\right)
\frac{\delta H}{\delta {\bf m}}
-
\left(\frac{\delta H}{\delta {\bf m}}\,
\cdot\nabla\right)
\frac{\delta G}{\delta {\bf m}}
\right]\,{\rm d}^n{\bf x}
\\
-
\int
\boldsymbol{C}\cdot\left[
\left(
\frac{\delta G}{\delta {\bf m}}\,
\cdot\nabla\right)
\frac{\delta H}{\delta \boldsymbol{C}}
-
\left(\frac{\delta H}{\delta {\bf m}}\,
\cdot\nabla\right)
\frac{\delta G}{\delta \boldsymbol{C}}
\right]\,{\rm d}^n{\bf x}
\,-\,
\bigg\langle \boldsymbol{C},\,\left[
\frac{\delta G}{\delta \boldsymbol{C}},\,
\frac{\delta H}{\delta \boldsymbol{C}}
\right]_\mathfrak{\!g}\bigg\rangle
\end{multline*}
}   

\section{Conclusions and open questions} We have shown how continuum equations on a certain class of semidirect-product Lie groups allow for singular solution momentum maps arising from the left action of diffeomorphisms on the $G$-bundle ${\rm Emb}(S,\mathbb{R}^n)\times
G(S)$. On the other hand, the right action on ${\rm Emb}\times G$ has been shown to yield another momentum map that recovers the Kelvin-Noether theorem.

These results arose from the observation that the collective
dynamics on ${\rm Emb}\times G$ was generated by a Kaluza-Klein Hamiltonian, thereby recovering the conservation of a gauge charge from a cyclic coordinate in the gauge group $G$.

The Kaluza-Klein construction for the collective motion was implemented in the continuum description by considering the semidirect product ${\rm Diff}\,\circledS\,G$ as the product of the diffeomorphisms with a gauge group $G$. The Kaluza-Klein construction implies the Kelvin-Noether theorem.

An important open question is whether the singular solutions (\ref{SDsingsoln}) and (\ref{singsolnansatz}) emerge  from smooth initial conditions. For example, the Camassa-Holm equation shows the spontaneous emergence of singular solutions from {\it any} smooth initially confined velocity configuration \cite{CaHo1993} and it is reasonable to ask whether this feature is shared by the class of  two-component Abelian and non-Abelian Camassa-Holm systems considered here. This question is being pursued elsewhere \cite{HoONaTr2008}.

Another open question concerns more general semidirect products.
In fact, the present discussion has considered only semidirect products of the Diff group with $G$-valued scalar functions. However, in physical applications one may also find semidirect products of the form ${\rm Diff}\,\circledS\,{\cal T}$, where $\cal T$ denotes tensor fields in physical space. The most important example is probably ideal magnetohydrodynamics, where ${\cal T}$ is the space of exact two forms (cf. e.g. \cite{HoMaRa}). The dynamics on such products differs substantially from the cases considered here and deserves further investigation.

\bigskip
\subsection*{Acknowledgements}
We are grateful to David Ellis, Fran\c{c}ois Gay-Balmaz, Andrea
Raimondo, Tudor Ratiu, Alain Trouv\'e and Laurent Younes for stimulating discussions. This work was partially supported by  the Royal Society of London Wolfson Research Merit Award.

\bibliographystyle{unsrt}

\end{document}